\documentclass[twocolumn, amssymb, amsmath, showpacs,pra]{revtex4}
\usepackage{amssymb}
\usepackage{mathrsfs}
\usepackage{}

\usepackage{graphicx,graphics,epsfig,subfigure,times,bm,bbm,amssymb,amsmath,amsfonts,amsthm,mathrsfs,MnSymbol}
\usepackage[matrix,frame,arrow]{xypic}
\usepackage[pdfstartview=FitH]{hyperref}
\usepackage[pdftex]{color}
\usepackage{hypernat}
\usepackage{enumerate}
\usepackage{epstopdf}
\usepackage{pdfsync}

\input{Qcircuit.tex}

\newtheorem{definitionenv}{Definition}
\newtheorem{remarkenv}[definitionenv]{Remark}
\newenvironment{remark}{\begin{remarkenv}\rm}{\end{remarkenv}}

\newtheorem{thm}{Theorem}
\newtheorem{mydef}{Definition}

\newtheorem{mylemma}{Lemma}

\newtheorem{myproposition}{Proposition}

\newcommand{\bes} {\begin{subequations}}
\newcommand{\ees} {\end{subequations}}
\newcommand{\bea} {\begin{eqnarray}}
\newcommand{\eea} {\end{eqnarray}}

\newcommand{\beq}{\begin{equation}}

\newcommand{\eeq}{\end{equation}}

\newcommand{\ignore}[1]{}

\def\G{\Gamma}

\def\>{\rangle}
\def\<{\langle}
\def\Tr{\mathrm{Tr}}

\newcommand{\multistate}[2]{*+{\hphantom{#2}} \POS[0,0].[#1,0] !C
  *{#2} \POS[0,0].[#1,0] \drop\frm{}}
\newcommand{\ghoststate}[1]{*+{\hphantom{#1}} }
\newcommand{\ccteq}[1]{\multistate{#1}{=}}
\newcommand{\ccteqg}{\ghoststate{=}}

\begin{document}
\title{A Fault-Tolerant Scheme of Holonomic Quantum Computation on Stabilizer Codes with Robustness to Low-weight Thermal Noise}
\author{Yi-Cong Zheng}

\email{
yicongzh@usc.edu
}
\author{Todd A. Brun}
\email{tbrun@usc.edu}
\affiliation{Department of Electrical Engineering, Center for Quantum Information Science \& Technology, University of Southern California, Los Angeles, California, 90089\\}
\date{\today}
\begin{abstract}
We show an equivalence relation between fault-tolerant circuits for a stabilizer code and fault-tolerant adiabatic processes for holonomic quantum computation (HQC), in the case where quantum information is encoded in the degenerated ground space of the system Hamiltonian. By this equivalence, we can systematically construct a fault-tolerant HQC scheme, which can geometrically  implement a universal set of encoded quantum gates by adiabatically deforming the system Hamiltonian. During this process, quantum information is protected from low weight thermal excitations by an energy gap that does not change with the problem size.
\end{abstract}
\pacs{03.65.Vf, 03.67.Lx, 03.67.Pp}
\maketitle

\section{Introduction}
Quantum computers are superior to classical ones in solving specific difficult problems, yet they are extremely vulnerable to errors during the computation process. It has been shown that if the errors of each type are local, and their rates are below a certain threshold, it is possible to implement quantum algorithms reliably with arbitrarily small error ~\cite{Nielsen:2000:CambridgeUniversityPress,DivencenzoFTPhysRevLett.77.3260,Kitaev:2003:2,QECbook:2013}. These quantum threshold theorems are based on the idea of quantum error correction (QEC).

In addition to QEC, there have also been proposals to deal with noise by designing the ``hardware" to provide inherent robustness. One of such method is holonomic quantum computing (HQC)~\cite{Zanardi:1999:94}---an all-geometric, adiabatic method of computation that uses a non-Abelian generalization of the Berry phase~\cite{Wilczek:1984:2111}. This approach is robust against certain types of errors during the adiabatic evolution~\cite{Solina:robustofHQCPhysRevA.70.042316,Gurdi:robustofHQCPhysRevLett.94.020503,solinas2012stability_HQC} and offers some built-in resilience.

Another method is to use the adiabatic quantum computing (AQC)~\cite{Farhi:0001106,FarhiScience} model instead of the standard quantum circuit model, which slowly drags the ground state of the system to the final Hamiltonian, whose ground state encodes the solution of the problem. AQC would take advantage of the energy gap between the ground state and other excited states to suppress thermal noise when evolution is very slow~\cite{Jordan:2005:052322,TameemNJP:adiabaticMarkovianME}.

The combination of fault-tolerant techniques and HQC was studied in Ref.~\cite{OgyanHolonomicQCPhysRevLett.102.070502,OgyanHolonomicQcPhysRevA.80.022325}, where the system Hamiltonian is an element of the stabilizer group or gauge group. Single qubit or two-qubit unitary operations are realized through continuously deforming the the system Hamiltonian. During this process, the path in the parameter space forms an open loop and results in the desired unitary transformation. After a sequence of such elementary operations, a closed-loop holonomy is obtained in the code space. However, this approach does not protect quantum information from thermal noise since there is no energy gap between the code space and error spaces. Also, while considerable work has been done in \cite{Jordan:2005:052322,LidarTowardFTAdqcPhysRevLett.100.160506},
a fault-tolerant theory for AQC is still lacking. The system's minimal energy gap, which determines the time scale of evolution, scales as an inverse polynomial in the problem size~\cite{StewartEquivalenceADCPhysRevA.71.062314,AriEquivalenceADCPhysRevLett.99.070502}, so that the temperature must be lowered polynomially to prevent thermal excitation.

In this paper, we present a scheme combining advantages of all three methods mentioned previously. First, we show an equivalence relation between fault-tolerant circuits and fault-tolerant adiabatic processes in the case where quantum information is encoded in a code space, which is also the ground space of a system Hamiltonian. Based on this, we present an alternative way to systematically construct a fault-tolerant HQC process that takes advantage of the energy gap between the ground space and other excitation states. Unlike AQC, this gap does not change with the problem size, and we know the exact value of the gap during the process, which greatly enhances the ability to prevent low-weight thermal excitation. With a lower error rate at the physical level of the fault-tolerant scheme, it may help to reduce the number of qubits needed and the frequency of error detection and error correction.

The structure of this paper is as follows. In Sec.~\ref{sec:math_framework}, we review the preliminaries that we will use to formulate our problem. Specifically, after defining HQC in Sec.~\ref{sec:hc_general}, we review the geometrical setting of the holonomic problem in Sec.~\ref{sec:hc_math}, and the basic ideas of fault-tolerant quantum computing in Sec.~\ref{sec:code_and_fault_tolerant}.
We connect fault-tolerant techniques and HQC in  Sec.~\ref{sec:fault_tolerant_hqc}. In Sec.~\ref{sec:scheme}, we describe our method to construct an adiabatic process from a fault-tolerant circuit to implement encoded unitary operations. Then in Sec.~\ref{sec:fault_tolerant}, we prove that our method of constructing encoded unitary operations is fault-tolerant, and discuss how it can realize universal fault-tolerant quantum computation. Several examples are given in Sec.~\ref{sec:examples}. In Sec.~\ref{sec:3qubit}, we show how our scheme works on the simplest 3-qubit repetition code. A less trivial example, of the encoded CNOT gate for the Steane code, is given in Sec.~\ref{sec:Steane's_code}. We summarize our results and conclude in Sec.~\ref{sec:conclusion}.

\section{Preliminaries}\label{sec:math_framework}

\subsection{Holonomic quantum computation}\label{sec:hc_general}

Consider a family of Hamiltonians $\{H_\lambda\}$ on an $N-$dimensional Hilbert space. The point $\lambda$, parametrizing the Hamiltonian, is an element of a manifold $\mathcal {M}$ called the control manifold, and the local coordinates of $\lambda$ are denoted by $\lambda^i\ (1\leq i \leq \textrm{dim}\mathcal{M})$. Assume there are only a fixed number of eigenvalues $\varepsilon_k(\lambda)$ (this is the case we are interested in) and suppose the $n$th eigenvalue $\varepsilon_n(\lambda)$ is $K_n$-fold degenerate for any $\lambda$. The degenerate subspace at $\lambda$ is denoted by $\mathcal{H}_n(\lambda)$. The orthonormal basis vectors of $\mathcal{H}_n(\lambda)$ are denoted by $\{|\phi_\alpha^n;\lambda\>\}$, satisfying
\beq
H_\lambda|\phi_\alpha^n;\lambda\>=\varepsilon_n(\lambda)|\phi_\alpha^n;\lambda\>,
\eeq
and
\beq
\<\phi_\alpha^n;\lambda|\phi_\beta^m;\lambda\>=\delta_{nm}\delta_{\alpha\beta}.
\eeq
Now assume the parameter $\lambda$ is changed adiabatically, which means that
\beq\label{eq:adiabatic_approx}
\left(\varepsilon_n(\lambda(t))-\varepsilon_{n^\prime}(\lambda(t))\right)T\gg1
\eeq
is satisfied for $n\neq n^\prime$ during $0\leq t\leq T$). Suppose the initial state at $t=0$ is an eigenstate $|\psi^n(0)\>=|\phi_\alpha^n;\lambda(0)\>$. The Schr\"{o}dinger equation is
\beq \label{eq:shrodinger}
i\frac{\text{d}}{\text{d}t}|\psi^n(t)\>=H(\lambda(t))|\psi^n(t)\>,
\eeq
whose solution will have the form
\beq \label{eq:solution}
|\psi^n(t)\>=\sum_{\beta=1}^{K_n}|\phi_\beta^n;\lambda(t)\>U_{\beta\alpha}(t).
\eeq
where we have used the adiabatic approximation from Eq.~(\ref{eq:adiabatic_approx}). Substituting Eq.~(\ref{eq:solution}) into Eq.~(\ref{eq:shrodinger}), one finds that $U_{\beta\alpha}$ satisfies
\beq
\begin{split}
\dot{U}_{\beta\alpha}(t)=&-i\varepsilon_n(\lambda(t))U_{\beta\alpha}(t)\\
&-\sum_{\mu}\<\phi_\beta^n;\lambda(t)|\frac{\text{d}}{\text{d}t}|\phi_\mu^n;\lambda(t)\>U_{\mu\alpha}(t).
\end{split}
\eeq
The solution can be expressed as
\beq
U(t)=\exp\left(-i\int_0^t\varepsilon_n(\lambda(s))\text{d}s\right)\mathcal{T}\exp\left(-\int_{0}^{t} A^n(\tau)\text{d}\tau\right),
\eeq
where $\mathcal {T}$ is the time-ordering operator and
\beq\label{eq:WZ_connection}
A^n_{\beta\alpha}(t)=\<\phi^n_{\beta};\lambda(t)|\frac{\text{d}}{\text{d}t}|\phi^n_\alpha;\lambda(t)\>
\eeq
is the Wilczek-Zee (WZ) connection~\cite{Wilczek:1984:2111}. Define the connection
\beq
\mathcal{A}^n_{i,\beta\alpha}(t)=\<\phi^n_{\beta};\lambda(t)|\frac{\partial}{\partial \lambda^i}|\phi^n_\alpha;\lambda(t)\>,
\eeq
through which $U(t)$ can be expressed as
\beq
U(t)=\exp\left(-i\int_0^t\varepsilon_n(\lambda(s))\text{d}s\right)\mathcal{P}\exp\left(-\int_{\lambda(0)}^{\lambda(t)}\sum_i\mathcal{A}^n_i\text{d}\lambda^i\right),
\eeq
where $\mathcal{P}$ is the path-ordering operator. Suppose the path $\lambda(t)$ is a loop $\lambda$ in $\mathcal{M}$ such that $\lambda(0)=\lambda(T)=\lambda_0$. Then after transporting through $\lambda$, states are transformed to
\beq
|\psi^n(T)\>=\sum_{\beta=1}^{K_n}|\psi_\beta^n(0)\>U_{\beta\alpha}(T).
\eeq
The unitary matrix
\beq
\Gamma_\lambda
=\mathcal{P}\exp\left(-\oint_\lambda\sum_i\mathcal{A}^n_i\text{d}\lambda^i\right)
\eeq
is called the holonomy associated with the loop $\lambda(t)$. $\Gamma_\lambda$ is a purely geometric object, and is independent of the parametrization of the path. Note that for a given $\Gamma_\lambda$, there exist infinitely many paths $\lambda$. Given a path $\lambda$, to find the holonomy is easy. However, the inverse problem---given a holonomy, to find the the proper path $\lambda$---is in general not trivial at all. In the rest of the paper, we will discuss how to find a proper path $\lambda$ to realize a certain holonomy in the code space, and thus perform an encoded quantum gate operation.

\subsection{Formulation of geometric problem}\label{sec:hc_math}
The definition introduced in Sec.~\ref{sec:hc_general} is not easy to use for our purpose. In this section, we outline the geometric setting of the holonomic problem as described in Refs.~\cite{Tanimura2004199,tanimura:022101}, which gives a clearer picture and more concise formulation of the problem. We focus on the ground space of the Hamiltonian to simplify the problem. However, this formalism is quite general, and can be applied to any eigenspace of the system Hamiltonian.

Suppose we have a family of Hamiltonians acting on the Hilbert space $\mathbb{C}^N$, and the ground state of each Hamiltonian is $K$-fold degenerate ($K<N$). The natural mathematical setting to describe this system is the principal bundle $(S_{N,K}(\mathbb{C}), G_{N,K}(\mathbb{C}),\pi, {\rm U}(K))$, which consists of the Stiefel manifold $S_{N,K}(\mathbb{C})$, the Grassmann manifold $G_{N,K}(\mathbb{C})$, the projection map $\pi:S_{N,K}(\mathbb{C})\rightarrow G_{N,K}(\mathbb{C})$, and the unitary structure group ${\rm U}(K)$. We explain the meaning of these mathematical objects below.

The Stiefel manifold is defined as:
\beq
S_{N,K}(\mathbb{C})=\{V\in M(N,K;\mathbb{C})|V^\dag V=I_K\},
\eeq
where $M(N,K;\mathbb{C})$ is the set of $N\times K$ complex matrices and $I_K$ is the $K-$dimensional unit matrix. Physically, each column of $V\in S_{N,K}(\mathbb{C})$ can be viewed as a normalized state in $\mathbb{C}^N$, and $V$ can be viewed as an orthonormal set of $K$ basis of the ground space of Hamiltonian. Since we have freedom to transfer from one orthnormal basis of ground space to another through unitary transformation, we can then define the unitary group $\textrm{U}(K)$ that acts on $S_{N,K}(\mathbb{C})$ from the right:
\beq\label{eq:gauge_transform}
S_{N,K}(\mathbb{C})\times\textrm{U}(K)\rightarrow S_{N,K}(\mathbb{C}), \quad (V,h)\mapsto Vh,
\eeq
by the matrix product of $V$ and $h$. $V$ and $Vh$ can be viewed as two
different orthonormal basis corresponding to the ground space.

During the adiabatic evolution, the ground space of the Hamiltonian may change. The ground space can be viewed as a $K$-dimensional hyperplane in $\mathbb{C}^N$. So we introduce the Grassmann manifold in $\mathbb{C}^N$:
\beq
G_{N,K}(\mathbb{C})=\{P\in M(N,N;\mathbb{C})|P^2=P,P^\dag=P,\text{Tr}P=K\},
\eeq
where $P$ is a projection operator onto the hyperplane in $\mathbb{C}^N$, and the condition $\text{Tr}P=K$ indicates that the hyperplane is $K$-dimensional. In our scenario, $P\in G_{N,K}(\mathbb{C})$ can be regarded as the projector onto the $K$-dimensional ground space of the Hamiltonian.

The relationship of the orthonormal basis $V$ and ground space $P$ can be seen as follows. We define the projection map $\pi:S_{N,K}(\mathbb{C})\rightarrow G_{N,K}(\mathbb{C})$ as
\beq
\pi:V\mapsto P:=VV^\dag.
\eeq
So, given an orthonormal basis, we can obtain the corresponding ground space projector.
We can check that the basis $V$ and basis $Vh$ with $h\in \text{U}(K)$ belong to the same ground space, since
\beq\label{eq:projection}
\pi(Vh)=(Vh)(Vh)^\dag=Vhh^\dag  V^\dag = VV^\dag = \pi(V).
\eeq


In our scenario of HQC, we want to transform the ground space adiabatically during the procedure. To formulate such a process, we also define the left action of the unitary group ${\rm U}(N)$ on both $S_{N,K}(\mathbb{C})$ and $G_{N,K}(\mathbb{C})$ by the matrix product:
\beq\label{eq:basis_transform}
{\rm U}(N)\times S_{N,K}(\mathbb{C})\rightarrow S_{N,K}(\mathbb{C}),\quad (g,V)\mapsto gV,
\eeq
and
\beq\label{eq:space_transform}
{\rm U}(N)\times G_{N,K}(\mathbb{C})\rightarrow G_{N,K}(\mathbb{C}),\quad (g,P)\mapsto gPg^\dag.
\eeq
It is easy to see that $\pi(gV)=g\pi(V)g^\dag$. This action is transitive: there is a $g\in{\rm U}(N)$ for any $V,V^\prime\in S_{N,K}(\mathbb{C})$ such that $V^\prime=gV$. There is also a $g\in{\rm U}(N)$ for any $P,P^\prime\in G_{N,K}(\mathbb{C})$ such that $P^\prime=gPg^\dag$. So this action is sufficient to describe any ground space transformation.



The canonical connection form on $S_{N,K}(\mathbb{C})$
is defined as a $\mathbbm{u}(K)$-valued one-form on $G_{N,K}(\mathbb{C})$:
%
\beq
A=V(P)^\dag \text{d}V(P),
\eeq
which is a generalization of the WZ connection in Eq.~(\ref{eq:WZ_connection}). This is the unique connection that is invariant under the transformation in Eq.~(\ref{eq:gauge_transform}):
\beq
\begin{split}
\tilde{A}=&h^\dag V(P)^\dag \text{d}\left(V(P)h\right)\\
=&h^\dag A h + h^\dag \text{d}h,
\end{split}
\eeq
which can be recognized as a gauge transformation.

Now we apply this formalism to the system dynamic of HQC. The state vector $\psi(t)\in\mathbb{C}^N$ evolves according to the Schr\"{o}dinger equation,
\beq\label{eq:shrodinger_2}
i\frac{\text{d}}{\text{d}t}\psi(t)=H(t)\psi(t).
\eeq
The Hamiltonian has a spectral decomposition,
\beq
H(t)=\sum_{l=0}^L\varepsilon_l(t)P_l(t),
\eeq
with projection operators $P_l(t)$. Therefore, the set of energy eigenvalues $(\varepsilon_0(t),\ldots,\varepsilon_L(t))$ and orthogonal projectors $(P_0(t),\ldots, P_l(t))$ encodes the information of the control parameters of the system. For the ground space, without loss of generality, the energy is assumed to be zero: $\varepsilon_0(t)=0$. We write $P_0(t)$ as $P(t)$ for simplicity. Suppose the degeneracy $K=\text{Tr}P(t)$ is constant. For each $t$, there exists $V(t)\in S_{N,K}(\mathbb{C})$ such that $P(t)=V(t)V^\dag(t)$. By the adiabatic approximation, we can substitute for $\psi(t)\in\mathbb{C}^N$ a reduced state vector $\phi(t)\in \mathbb{C}^K$:
\beq
\psi(t)=V(t)\phi(t).
\eeq
Since $H(t)\psi(t)=\varepsilon_0(t)\psi(t)=0$, the Schr\"{o}dinger equation ~(\ref{eq:shrodinger_2}) becomes
\beq
\frac{\text{d}\phi}{\text{\text{d}}t}+V^\dag\frac{\text{d}V}{\text{d}t}\phi(t)=0,
\eeq
and the solution can be represented formally as
\beq
\phi(t)=\mathcal {P}\exp\left(-\int V^\dag \text{d}V\right)\phi(0).
\eeq
Therefore, $\psi(t)$ can be written
\beq\label{eq:adiabatic_evolution}
\psi(t)=V(t)\mathcal{P}\exp\left(-\int V^\dag \text{d}V\right)V^\dag(0)\psi(0).
\eeq
In particular, if the system comes back to its initial point, as $P(T)=P(0)$, the holonomy $\Gamma\in\text{U}(K)$ is defined as
\beq\label{eq:holonomy}
\Gamma=V^\dag(0) V(T)\mathcal{P}\exp\left(-\int V^\dag \text{d}V
\right),
\eeq
and the final state is
\beq
\begin{split}
\psi(T)&=V(0)\Gamma\phi(0).\\
\end{split}
\eeq
According to the formula above, an operation $\Gamma\in {\rm U}(K)$ is applied to the ground space.

If the condition
\beq\label{eq:horizontal_condition}
V^\dag\cdot\frac{\text{d}V}{\text{d}t}=0
\eeq
is satisfied for all $t$, the curve $V(t)$ in $S_{N,K}(\mathbb{C})$ is called a horizontal lift of the curve $P(t)=\pi(V(t))$ in $G_{N,K}(\mathbb{C})$.Then the holonomy ~(\ref{eq:holonomy}) is greatly simplified to
\beq\label{eq:horizontal_holonomy}
\Gamma=V^\dag(0)\cdot V(T)\in\text{U}(K).
\eeq

Now we are ready to reformulate the problem stated at the end of Sec.~\ref{sec:hc_general}. Given a desired unitary operation  $U_{\text{op}}\in \text{U}(K)$ and a fixed initial point $P(0)\in G_{N,K}(\mathbb{C})$, we want to find a loop $P(t)\in G_{N,K}(\mathbb{C})$ with base points $P(0)=P(T)$ whose horizontal lift $V(t)\in S_{N,K}(\mathbb{C})$ produces the holonomy $\G = U_{\text{op}}$ according to Eq.~(\ref{eq:horizontal_holonomy}). In Sec.~\ref{sec:scheme}, we will discuss in detail how to find such a loop $P(t)$ whose horizontal lift gives the desired holonomy in the code space.

\begin{figure}[!htp]
\includegraphics[width=80mm]{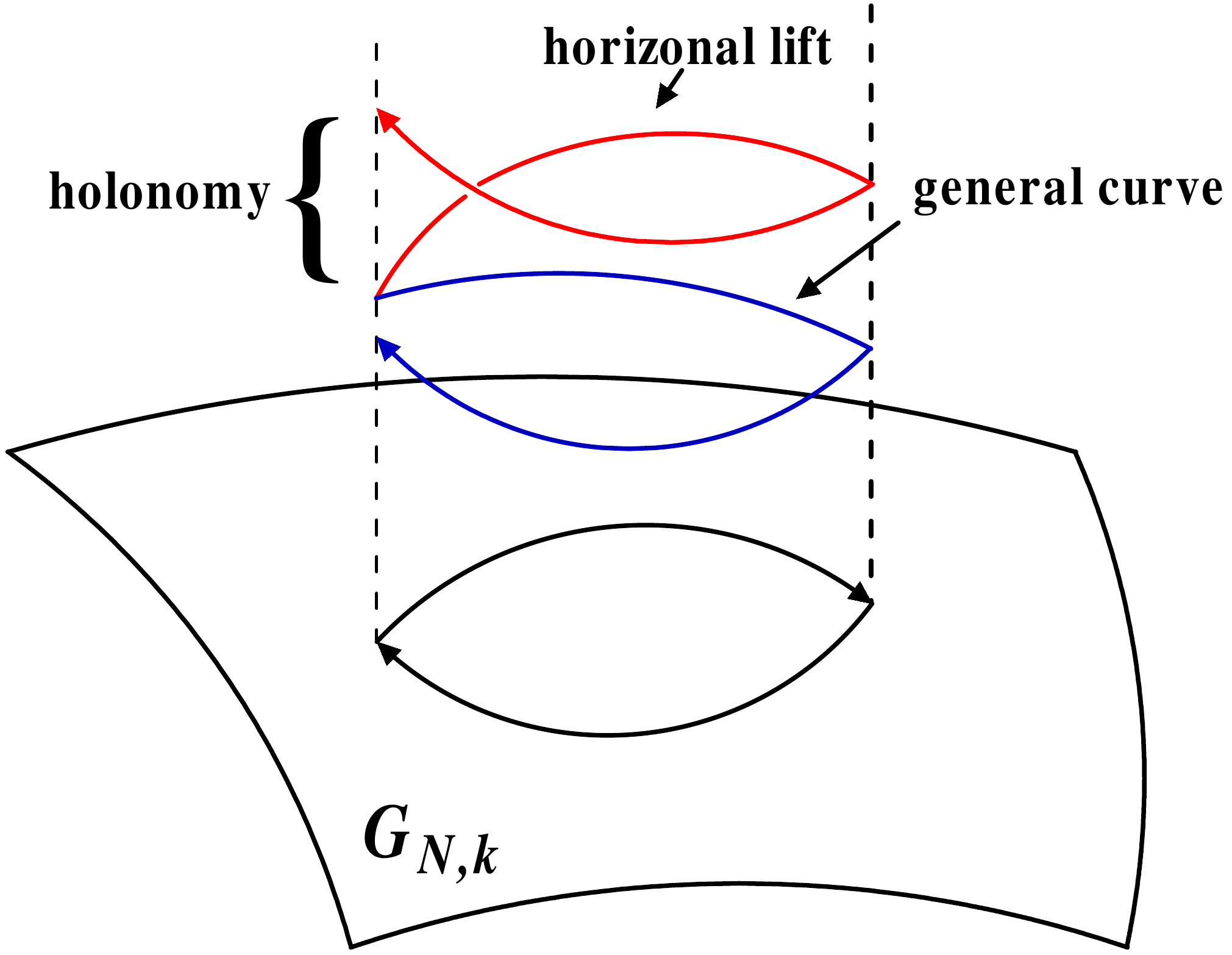}
\caption{\label{Fig:horizontal_lift}Horizontal lift as a specified curve in $S_{N,K}(\mathbb{C})$ whose projection is $P(t)$. The initial condition $V(0)$ becomes $V(T)$, which is generally different from $V(0)$. The difference is the holonomy.}
\end{figure}
A visualization of a horizontal lift is shown in Fig.~\ref{Fig:horizontal_lift}. Without loss of generality, we can always restrict ourselves to the case where $P(t)$ has the form
\beq
P(t)=U(t,0)P(0)U^{\dag}(t,0)=U(t,0)v_0v_0^\dag U^{\dag}(t,0),
\eeq
for some smooth $U(t,0)\in {\rm U}(N)$ according to Eq.~(\ref{eq:space_transform}). Note that, $U(t,0)$ should be chosen such that, at any time $t$
\beq
U(t+\tau,t)P(t)U^{\dag}(t+\tau,t)\neq P(t),
\eeq
for some neighborhood of $t$. This condition can also be stated as
\beq\label{eq:eff_evolution}
\left[\frac{\partial}{\partial\tau}U(t+\tau,t)|_{\tau=0},P(t)\right]\neq 0.
\eeq
The case where Eq.~(\ref{eq:eff_evolution}) equals 0 is allowed only at a finite number of points in $[0,T]$. The horizontal curve should satisfy the following set of equations:
\beq\label{eq:horizontal_equation_group}
\begin{split}
    V^\dag\cdot\frac{\text{d}V}{\text{d}t}&=0, \\
    P(t)=V(t)V^\dag(t)&= U(t,0)v_0v_0^\dag U^{\dag}(t,0).
\end{split}
\eeq
The general solution to these equations can
be written as:
\beq\label{eq:general_curve}
V(t)=U(t,0)v_0h(t,0)
\eeq
for some $h(t,0)\in {\rm U}(K)$. Substituting Eq.~(\ref{eq:general_curve}) into Eq.~(\ref{eq:horizontal_equation_group})
we get:
\beq\label{eq:ht_evolution}
\dot{h}(t,0)=-v_0^\dag U^{\dag}(t,0)\dot{U}(t,0)v_0 h(t,0).
\eeq

A well known result of differential geometry about the uniqueness of a horizontal lift curve~\cite{Nakahara:2003:IOP} can now be directly proved in this specific scenario, which will be used later.
\begin{mylemma}\label{lemma:horiziontal_uniqueness}
Let $P:[0,T]\rightarrow G_{N,K}(\mathbb{C})$ be a curve in $G_{N,K}(\mathbb{C})$ and let $v_0\in\pi^{-1}(P(0))$. Then there exists a unique horizontal lift $V(t)$ in $S_{N,K}(\mathbb{C})$ such that $V(0)=v_0$ .
\end{mylemma}
\begin{proof}
It's easy to show that $v_0^\dag U^{-1}(t,0)\dot{U}(t,0)v_0$ is anti-Hermitian, so $h(t,0)\in {\rm U}(K)$ for all $t$. Define $V^\prime(t)=U(t,0)V(0)$ to be a particular curve in principal bundle that gives a corresponding WZ connection $A^\prime=V^{\prime\dag} \text{d}V^\prime$.  With initial
condition $h(0,0)=I_K$, the solution of Eq.~(\ref{eq:ht_evolution}) can be written as:
\beq
h(t,0)=\mathcal{P}\exp\left(-\int A^\prime\right),
\eeq
and hence there exists an unique horizontal lift $V(t)$.
\end{proof}

\subsection{Stabilizer codes and fault-tolerant computation}\label{sec:code_and_fault_tolerant}
A quantum error-correcting code is formally defined as a subspace $\mathcal{C}$ of some larger Hilbert space. A necessary and sufficient condition for a set of errors $\{E_i\}$ to be correctable is \cite{Nielsen:2000:CambridgeUniversityPress,QECbook:2013}:
\beq\label{eq:error_correction_condition}
PE_i^\dag E_jP=\alpha_{ij}P, \quad \forall i,j,
\eeq
for some Hermitian matrix $\alpha$. Here $P$ is the projector onto $\mathcal{C}$. Since any linear combination of $\{E_i\}$ is also correctable, we define
\beq
\mathcal{E}=\text{Span}\{E_i\}
\eeq
to be a correctable error set for code $\mathcal{C}$.

The codes we are interested in are the stabilizer codes \cite{Gottesman:9705052}. We briefly review the formalism of stabilizer codes. Let $G_n$ be the Pauli group acting on $n$ qubits. An Abelian subgroup $\mathcal{S}$ of $G_n$ is called a stabilizer group if $-I\nin\mathcal{S}$. The stabilizer group defines a subspace of the $n$-qubit Hilbert space by
\beq
\mathcal{C}=\{|\psi\>:S|\psi\>=|\psi\> \textrm{ for all } S\in \mathcal{S}\}.
\eeq
This $\mathcal{C}$ is called the code space. $\mathcal{C}$ is nonzero since $-I\nin \mathcal{S}$. A state in $\mathcal{C}$ is called a codeword. This subspace is the simultaneous +1 eigenspace of the stabilizer group. If the subspace has dimension $2^k$ ($k$ logical qubits), the stabilizer code can be specified by $n-k$ commuting stabilizer generators, which are elements of $G_n$. The group $\mathcal{S}$ can be represented by these stabilizer generators: $\mathcal{S}=\<\{S_j\}\>$. All stabilizer codes can be characterized by three parameters $[[n,k,d]]$, where $d$ is the minimum distance of the code, which is equal to the
minimum weight of all nontrivial elements in the normalizer group of $\mathcal{S}$.

With the use of stabilizer codes, it is possible to build a quantum processor that is fault-tolerant~\cite{Nielsen:2000:CambridgeUniversityPress,QECbook:2013,Gottesman:9705052}.
A quantum information processor is called fault-tolerant if the information is encoded in a quantum error-correcting code at all times during the procedure, and a failure at any point in the procedure can only propagate to a small number of qubits, so that error correction can remove the errors.
It has been shown that fault-tolerant computation is possible on any stabilizer code~\cite{QECbook:2013,Gottesman:9705052} for some error model. Typically, there are three elementary quantum ``gadgets": encoded state preparation, encoded unitary operations and encoded state measurement. Through enlarging or concatenating the fault-tolerant gadgets, a computation can achieve arbitrary accuracy, if the error rate is low enough~\cite{QECbook:2013}. Encoded Clifford unitary operations play a key role in fault-tolerant computation, since for most proposed schemes of fault-tolerant quantum computation, like concatenation of the Steane code~\cite{Nielsen:2000:CambridgeUniversityPress}, C4 code~\cite{KnillFTNature} or surface code~\cite{Folwer2012PhysRevA.86.032324}, we can prepare encoded non-Clifford magic states using techniques like state distillation, which can be implemented by encoded Clifford unitary operations. So we will focus on encoded Clifford operations and their holonomic implementation.

\begin{figure}[!ht]
\includegraphics[width=80mm]{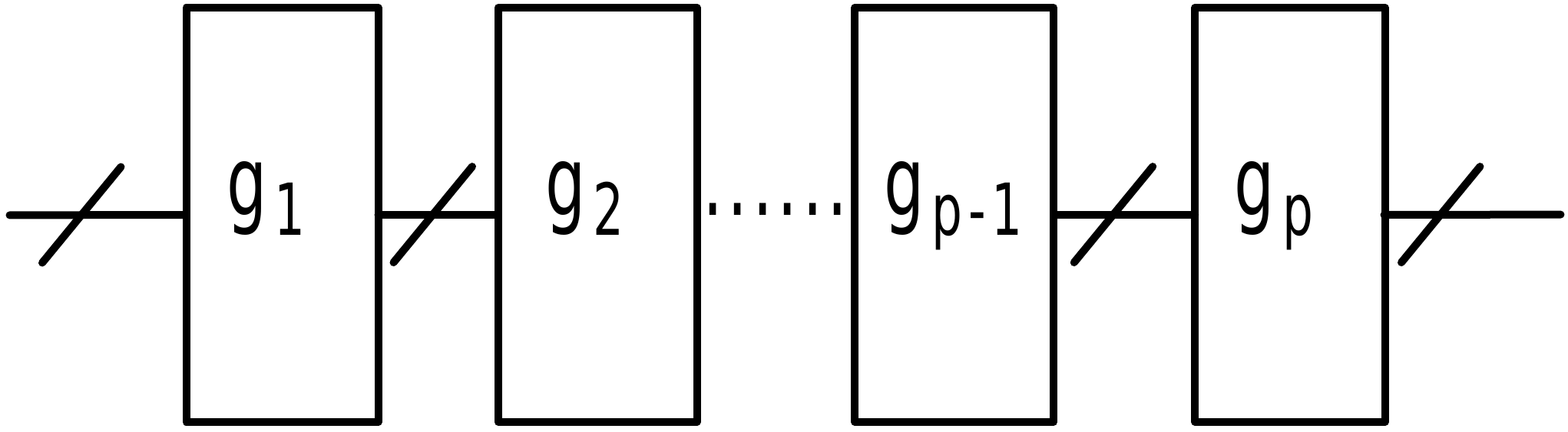}
\caption{\label{Fig:circuit} A logical unitary quantum operation is realized as a series of quantum gates from a universal set of gates in the circuit model.}
\end{figure}
In the standard circuit model, an encoded unitary operation can be realized by a series of quantum gates, say $p$ gates chosen from a universal set of gates as shown in Fig.~\ref{Fig:circuit}. Commonly, the universal set of gates $\mathscr{U}_1=\{\text{Hadamard}, \text{CNOT}, S, \pi/8\}$ is used to describe the circuit, which is good for certain fault-tolerant schemes. Here, we choose another universal set:
\beq
\mathscr{U}_2=\Big\{R^x=\exp\left(-i\frac{\pi}{4}X\right), R^{zz}=\exp\left(i\frac{\pi}{4}Z\otimes Z\right),
S, \pi/8\Big\},
\eeq
which proves to be much more suitable for our adiabatic scheme. Errors can occur anytime during the process, both between and during gate operations. Noisy gates are always equivalent to a perfect gate followed by an error operator, so we can just focus on the errors occurring between the gates. If an error $E^q\in\mathcal{E}$ occurs between gates $q-1$ and $q$, it will propagate to
\beq
E^{q\prime}=\prod_{l=q}^p g_{p+q-l}\cdot E^q\cdot\left(\prod_{l=q}^p g_{p+q-l}\right)^\dag.
\eeq
If a circuit is fault-tolerant, we can suppose that $E^{q\prime}$ would be still in the same correctable error set $\mathcal{E}$. According to this observation, we give a generalized definition of a fault-tolerant circuit for a code $\mathcal{C}$:
\begin{mydef}\label{def:fault_tolerant_circuit}
Given a code $\mathcal{C}$ and a particular correctable error set $\mathcal{E}$ for this code, a circuit $\mathcal{G}$ that realizes an encoded unitary operation is called a fault-tolerant circuit for $\mathcal{C}$ if for any $1\leq q< p$, $U_{qp}=\prod_{l=q}^p g_{p+q-l}$ maps any subset of $\mathcal{E}$ to another subset of $\mathcal{E}$.
\end{mydef}
This definition of a fault-tolerant circuit may be, however, too strong. In practice, it may be very difficult to find such a code and corresponding circuit. For a practical error model, strongly correlated errors happen with much lower probability than weakly correlated or local ones, so we will focus on local errors. For example, if our codes are stabilizer codes, the error set $\mathcal{E}_{\text{local}}$ can be spanned by Pauli operators with weight less than $\lfloor\frac{d-1}{2}\rfloor$, which occur with relatively high probability. If we limit ourselves to such a high-probability correctable local error set $\mathcal{E}_{\text{local}}$, then we get a weaker version of the definition of a fault-tolerant circuit:
\begin{mydef}
Given a stabilizer code [[n,k,d]] with a correctable error set $\mathcal{E}$, and there is a high-probability local error set $\mathcal{E}_{\text{local}}\subset\mathcal{E}$, a circuit $\mathcal{G}$ that realizes an encoded unitary operation is called a fault-tolerant circuit for this code if for any $1\leq l< p$, $U_{qp}=\prod_{l=q}^p g_{p+q-l}$ maps $\mathcal{E}_{\text{local}}$ to some subset of $\mathcal{E}$.
\end{mydef}
\begin{remark}
According to this definition, the encoded fault-tolerant unitary circuit does not necessarily need to be transversal, although the reverse is always true. If a circuit built of gates from $\mathscr{U}_1$ is fault-tolerant, then we can decompose its gates into gates from $\mathscr{U}_2$, and the new circuit we obtain is also fault-tolerant. So, in the rest of the paper, we assume that the given circuits are composed of gates from $\mathscr{U}_2$.
\end{remark}
\begin{remark}
We should mention here that in the following discussion, we only consider  circuits that contain no $\pi/8$ gates. In other words, we limit ourselves to Clifford circuits, since non-Clifford circuits will cause tremendous complexity. This restriction will be further discussed in Sec.~\ref{sec:fault_tolerant}. Fortunately, fault-tolerant encoded Clifford operations for stabilizer codes are made of Clifford circuits, that do \emph{not} contain $\pi/8$ gates. Encoded non-Clifford operations usually do contain $\pi/8$ gates. We will not directly implement encoded non-Clifford operations, but instead make use of magic state distillation, so this is not a serious restriction.
\end{remark}

\section{Fault-tolerant holonomic quantum computation}\label{sec:fault_tolerant_hqc}
To combine the advantages of holonomic quantum computation with fault-tolerant computation techniques, the basic idea is to obtain a holonomy on the code space, which is the ground space of the system Hamiltonian, during an adiabatic evolution. One must make sure that the encoded quantum information is protected by a suitable error-correcting code throughout the Hamiltonian deformation. For simplicity, we assume that error correction is applied at the end of the cyclic adiabatic evolution. However, this may not necessarily be true in practice. We require that an error occurring during the deformation be correctable at the end:
\begin{myproposition}\label{prop_correctable_error}
Given a code $\mathcal{C}$ with a correctable error set $\mathcal{E}$, suppose the initial state is $|\psi(0)\> \in \mathcal{C}$, and the deformation of the Hamiltonian is adiabatic. Then, in general, each eigenspace will undergo some transformation. Given a desired encoded operation (in our case, a holonomy) $\Omega_g$ on the code space,  suppose a series of errors $\{E^{t_i}\}$ occur at times $t_i$ during the evolution. Then $\{E^{t_i}\}$ is correctable only if the final state $|\psi(T)\>\propto E^f\Omega_g|\psi(0)\>$, for some $E^f\in \mathcal{E}$.
\end{myproposition}

In the case when $\{E^{t_i}\}$ is empty, the statement is obvious.
The fault-tolerance of this process is well defined in the case when $\{E^{t_i}\}$ just has one element, say $E^{t_1}$. Following the spirit of fault-tolerant quantum computation in the circuit model, we define fault-tolerance for an adiabatic process:
\begin{mydef}\label{def:adiabatic_fault_tolerance}
Given a code $\mathcal{C}$, defined by the ground space of an initial Hamiltonian with a correctable error set $\mathcal{E}$, a desired encoded operation (holonomy) $\Omega_g$, and an initial state $|\psi(0)\>\in\mathcal{C}$, the corresponding cyclic adiabatic process is called fault-tolerant if any $E^t\in \mathcal{E}_{\text{local}}$ occurring at time $t$ leads to a final state  $E^f\Omega_a|\psi(0)\>$ for some $E^f\in \mathcal{E}$.
\end{mydef}
Unlike AQC, we need to measure the stabilizer generators and do error correction after a single or multiple cycles of encoded operations. At those points, we turn off the Hamiltonian and apply a standard error correction procedure. If this scheme is robust to low-weight thermal noise, and also evolves slowly enough that the adiabatic error is well below the threshold (which we will examine in some detail in Sec.~\ref{sec:fault_tolerant}), then the frequency of error recovery operations can be greatly reduced.

We will show how to construct a fault-tolerant adiabatic process to do a holonomic encoded quantum unitary operation starting from a fault-tolerant circuit, and we will prove that such a process is fault-tolerant by Def.~\ref{def:adiabatic_fault_tolerance}.

\subsection{Scheme}\label{sec:scheme}
Given a stabilizer code with stabilizer group $\mathcal{S}$ for $n$ qubits ($2^n=N$), we set the system Hamiltonian at the very beginning to be
\beq\label{eq:hamiltonian}
H(0) = -\sum_j S_j.
\eeq
Thus the code space is the ground space of the Hamiltonian with dimension $K=2^k$. We deform the Hamiltonian as follows:
\beq\label{eq:hamiltonian_2}
\begin{split}
H(t)&=\sum_{j}C_j(t)S_j(t)\\
&=\sum_{j}C_j(t)U(t,0)S_jU^{\dag}(t,0),\\
\end{split}
\eeq
with $S_j(t)=U(t,0)S_jU^{\dag}(t,0)$, and $[S_i(t),S_j(t)]=0$ for all $i$, $j$. $C_j(t)\in [-1,0]$ is the weight of $S_j(t)$ which is assumed to be controllable.  The $\{S_j(t)\}$ can be viewed as a set of generators of an Abelian group, such as a stabilizer group.

 The Hamiltonian also has a spectral decomposition
\beq
H(t)=\sum_{\textbf{s}}\varepsilon_{\textbf{s}}(t)P_{\textbf{s}}(t).
\eeq
Here, the $\{P_{\textbf{s}}(t)\}$ are projectors onto the simultaneous eigenspace of all the $S_j(t)$, with eigenvalues:
\beq
\varepsilon_{\textbf{s}}(t)=\sum_{j}C_j(t)s_j,
\eeq
where the labels $s_j=\pm1$ form a vector:
\beq
\textbf{s}=\{s_1,s_2\ldots s_{n-k}\}.
\eeq
When the Hamiltonian changes, as shown previously in Eq.~(\ref{eq:space_transform}), the ground space will also evolve. This defines a time-dependent code $\mathcal{C}_t$ .
Let $P_0(t)=U(t,0)P_0(0)U^{\dag}(t,0)$ be the projector onto the ground space of the Hamiltonian $H(t)$ such that $s_j=1$ for all $j$. We emphasize that $U(t,0)$ should be chosen such that
\beq\label{eq:effect_evolution}
\left[\frac{\partial}{\partial \tau}U(t+\tau,t)|_{\tau=0}, P_{\textbf{s}}(t)\right]\neq 0 \quad \text{for all } \textbf{s},
\eeq
except for a finite set points $t$, so that the deformation procedure is non-trivial for all eigenspaces.

This method will work only if the adiabatic condition for each eigenspace  $P_{\textbf{s}}$ is satisfied, so that each eigenspace undergoes some non-trivial holonomy during the cyclic evolution, in case an error takes the system to $P_{\textbf{s}}$ during the process. The standard adiabatic condition~\cite{Messiah:1965:North} can be reformulated for the eigenspace  $\{P_{\textbf{s}_\alpha}\}$:
\beq\label{eq:adiabatic_condition_general}
\frac{\parallel P_{\textbf{s}_\alpha}(t)\frac{\partial}{\partial t}H(t)P_{\textbf{s}_\beta}(t)\parallel_1}{K\left(\varepsilon_{\textbf{s}_\alpha}(t)
-\varepsilon_{\textbf{s}_\beta}(t)\right)^2}\approx  0, \ \text{for any}\ \alpha\neq\beta.
\eeq
This must hold for all $t\in[0,T]$, where $\parallel\cdot\parallel_1$ is the trace norm $(\parallel A\parallel_1=\Tr\sqrt{A^\dag A})$.  For Hamiltonians of the form in Eq.~(\ref{eq:hamiltonian_2}), it is very likely that different $P_{\textbf{s}}(t)$'s share the same eigenvalues so the adiabatic condition would \emph{not} be directly satisfied. We will show later a systematic way to engineer the deformation procedure so that each eigenspace $P_{\textbf{s}}(t)$ satisfies this condition during the adiabatic process.

In addition, each eigenspace should undergo the same holonomy to satisfy Prop.~\ref{prop_correctable_error}. Let's see how it works. Define:
\beq\label{eq:zero_h_condition}
U^{\dag}(t,0)\dot{U}(t,0)=iQ(t,0),
\eeq
where $Q(t,0)$ is Hermitian. In order to obtain the same holonomy for each $P_\textbf{s}$, according to Eq.~(\ref{eq:ht_evolution}),
$P_{\textbf{s}}(0)Q(t,0)P_{\textbf{s}}(0)$ should be related to
$P_0(0) Q(t,0)P_0(0)$ in some way.
If we can make $P_{\textbf{s}}(0)Q(t,0)P_{\textbf{s}}(0)$ either equal to the zero matrix or proportional to $P_\textbf{s}(0)$ for all $\textbf{s}$, then the character of the horizontal lift of $P_\textbf{s}(t)$ is completely determined by $U(t,0)$.

Now we are ready to describe the scheme to construct a fault-tolerant adiabatic process for a holonomic unitary operation, starting from a fault-tolerant circuit. First, we divide the time of evolution $[0,T]$ into $p$ segments. The $l$th segment is $[t_{l-1},t_l]$, and we set $t_0=0$ and $t_p=T$. Given a fault-tolerant circuit $\mathcal{G}$ that realizes an encoded operation $\Omega_g=\prod_{l=1}^{p}g_{p-l+1}$, we can follow the steps listed below:

\begin{enumerate}
  \item Set $l=1$ and $t_0=0$.
  \item Check the number of $S_j(t_{l-1})$ such that $[S_j(t_{l-1}), g_l]\neq 0$. If this number is odd, go to step 3, else go to step 4.
  \item For the $l$th time segment $[t_{l-1},t_l]$, choose a unitary operator $U_l(t,t_{l-1})=g_l^{f_l(t)}$, with $f_l:[t_{l-1},t_l]\rightarrow[0,1]$ a monotonic smooth function with boundary conditions $f(t_{l-1})=0$ and $f(t_{l})=1$. We deform the Hamiltonian such that  $H(t)=U_l(t,t_{l-1})H(t_{l-1})U_l^{\dag}(t,t_{l-1})$ in the interval $[t_{l-1},t_l]$. All $S_j(t_{l-1})$ are replaced at $t_l$ by $S_j(t_l)=g_lS_j(t_{l-1})g_l^\dag$, and $H(t_l)=-\sum_{j}S_j(t_l)$. Then go to step 5.
  \item We need an additional operation to break the degeneracy in this case, in order that the adiabatic condition be satisfied for all $P_{\textbf{s}}(t)$. From those $S_j(t)$ such that $[S_j(t_{l-1}), g_l]\neq 0$, we arbitrarily select one element, say $S_b(t_{l-1})$, and change the Hamiltonian to $H(t_{l-1}^\prime)=H(t_{l-1})+ C_bS_b(t_{l-1})$. $C_b$ is a constant between 0 and 1; we will choose it to be 0.5. This procedure can be done arbitrarily fast, and it will not affect a state entirely contained in any $P_{\textbf{s}}$, so we can just set $t_{l-1}^\prime=t_{l-1}$. We choose $U_l(t,t_{l-1})=g_l^{f_l(t)}$ in this case, where $f_l:[t_{l-1},t_l^\prime]\rightarrow[0,1]$ is a monotonic smooth function with boundary conditions $f_l(t_{l-1})=0$ and $f_l(t_{l}^\prime)=1$. At time $t_l^\prime$, the Hamiltonian becomes $H(t_l^\prime)=-\sum_j S_j(t_l^\prime)+C_bS_b(t_l^\prime)$ with $S_j(t_l^\prime)=g_l S_j(t_{l-1})g_l^\dag$. Then we remove the additional term in the Hamiltonian, leaving $H(t_{l})=-\sum_j S_j(t_l)$, where $S_j(t_l)=S_j(t_l^\prime)$. Again, this can be done arbitrarily fast, so, we can set $t_l^\prime=t_l$. Go to step 5.
  \item If $l=p$, the process is finished. Else, set $l=l+1$ and go to step 2.
\end{enumerate}

First, we will prove that in the case where no error happens during the adiabatic evolution this process indeed gives an encoded operation $\Omega_g$ on the code space. For a circuit $\mathcal{G}$, we
define a set $\mathscr{T}(\mathcal{G})=\{Z_m\}\bigcup\{X_m\}\bigcup \{Z_{m_1}\cdot Z_{m_2}\}$, where $m$ ranges over all qubits in $\mathcal{G}$ and $m_1$, $m_2$ range over all pairs of qubits shared by any two-qubit gates in $\mathcal{G}$.
\begin{thm}\label{thm:theorem1}
Given a fault-tolerant circuit $\mathcal{G}$ defined for a stabilizer code $\mathcal{C}$, with $\mathcal{E}\supseteq\mathcal{E}_{\text{local}}\supset\mathscr{T}(\mathcal{G})$, then following the steps listed above we can perform a holonomic encoded operation $\Omega_g=\prod_{l=1}^{p}g_{p-l+1}$ for the code space $P_0$.
\end{thm}
\begin{proof}
It is easy to check that there is always a finite energy gap between $P_0(t)$ and any other $P_\textbf{s}(t)$ during the process. So if we choose the time scale properly, $P_0(t)$ can always satisfy the adiabatic condition. Consider the $q$th step of the implementation. If the $q$th gate is single qubit gate, it acts on some qubit $m$. If it is a two qubit gate, it acts on a pair of qubits $m_1$ and $m_2$.
According to Eq.~(\ref{eq:zero_h_condition}), for the $q$th step, we define
\beq
i\tilde{Q}(t,t_{q-1})=P_0(t_{q-1})U_q^{\dag}(t,t_{q-1})\dot{U}_q(t,t_{q-1})P_0(t_{q-1}).
\eeq
Assume we are in step 3 (the argument for step 4 is the same with a trivial modification). $U_q(t,t_{q-1})$ is chosen to be $g_q^{f_q(t)}$, where $g_q$ is one of the gates from $\mathscr{U}_2$. $U_q(t,t_q)$ can be represented explicitly for this set of gates as follows:
\beq
\begin{split}
U_q^{R^x_m}(t,t_{q-1})=&\exp{\left(-i\frac{\pi}{4}f_q(t)X_m\right)},\\
U_q^{Z_{m_1}Z_{m_2}}(t,t_{q-1})=&\exp{\left(i\frac{\pi}{4}f_q(t)Z_{m_1}Z_{m_2}\right)},\\
U_q^{S_m}(t,t_{q-1})=&\exp{\left(-i\frac{\pi}{4}f_q(t)Z_m\right)},\\
U_q^{\frac{\pi}{8}_m}(t,t_{q-1})=&\exp{\left(-i\frac{\pi}{8}f_q(t)Z_m\right)}.\\
\end{split}
\eeq
Define the code $\mathcal{C}_{q-1}$ as the ground space of $H(t_{q-1})$, i.e., the space projected onto by $P_0(t_{q-1})$, and assume  $V_0(t_{q-1})$ to be the horizontal lift of $P_0(t_{q-1})=V_0(t_{q-1})V_0^\dag(t_{q-1})$. According to Def.~\ref{def:fault_tolerant_circuit}, $\prod_{l=q}^pg_{p+q-l}$ maps  $\mathcal{E}_{\text{local}}$ to a subset of $\mathcal{E}$, which is the correctable error set of our code $\mathcal{C}$, so it's easy to check that $\mathcal{E}_{\text{local}}$ is a correctable error set for code $\mathcal{C}_{q-1}$, defined by $P_0(t_{q-1})$. Since $\mathscr{T}(\mathcal{G}) \subset \mathcal{E}_{\text{local}}$, according to Eq.~(\ref{eq:error_correction_condition}), we could have:
\beq
\begin{split}
&\tilde{Q}^{R^x_m}(t,t_{q-1})=\\
&-P_0(t_{q-1})\frac{\pi}{4}\dot{f}_q(t)X_mP_0(t_{q-1})=\alpha_1(t)P_0(t_{q-1}),\\
&\tilde{Q}^{Z_{m_1}Z_{m_2}}(t,t_{q-1})=\\
&P_0(t_{q-1})\frac{\pi}{4}\dot{f}_q(t)Z_{m_1}Z_{m_2}P_0(t_{q-1})=\alpha_2(t)P_0(t_{q-1}),\\
&\tilde{Q}^{S_m}(t,t_{q-1})=\\
&-P_0(t_{q-1})\frac{\pi}{4}\dot{f}_q(t)Z_mP_0(t_{q-1})=\alpha_3(t)P_0(t_{q-1}),\\
&\tilde{Q}^{\frac{\pi}{8}_m}(t,t_{q-1})=\\
&-P_0(t_{q-1})\frac{\pi}{8}\dot{f}_q(t)Z_mP_0(t_{q-1})=\alpha_4(t)P_0(t_{q-1}).
\end{split}
\eeq
It is easy to see that $\alpha_r(t)$, $r=$1, 2, 3, 4 are all real.

It is necessary to check that, for each step, Eq.~(\ref{eq:effect_evolution}) is satisfied. We have
\beq
P_0(t)=U_q(t,t_{q-1})P_0(t_{q-1})U_q^{\dag}(t,t_{q-1}).
\eeq
We will just show the case where an $R^x$ gate is applied at the $q$th stage; the calculations for other gates are just the same.
\beq\label{eq:effect_evolution_in_proof}
\begin{split}
&\left[\frac{\partial}{\partial \tau}U_q^{X_m}(t+\tau,t)|_{\tau=0},P_0(t)\right]\\
=&-i\frac{\pi}{4}\left[\dot{f}_q(t)X_m, P_0(t)\right]\\
=&-i\frac{\pi}{4}U_q^{X_m}(t,t_{q-1})\left[\dot{f}_q(t)X_m, P_0(t_{q-1})\right]U_q^{X_m\dag}(t,t_{q-1}).\\
\end{split}
\eeq
We multiply $P_0(t_{q-1})$ by $[X_m, P_0(t_{q-1})]$ and have
\beq
\begin{split}
&P_0(t_{q-1})\left[X_m, P_0(t_{q-1})\right]\\
=&P_0(t_{q-1})\left(X_mP_0(t_{q-1})-P_0(t_{q-1})X_m\right)\\
=&-P_0(t_{q-1})\left(\frac{4\alpha_1(t)}{\dot{f}_q(t)\pi}I+X_m\right)\neq0.
\end{split}
\eeq
So,  $\left[X_m, P_0(t_{q-1})\right]\neq0$, if $X_m\notin \<S_j(t_{q-1})\>$. This is indeed true in this case, since for a well-defined circuit, $X_m\notin \<S_j(t_{q-1})\>$. Otherwise, $R^x_m$ would have no effect at the $q$th stage. Then we have
\beq
\left[\frac{\partial}{\partial \tau}U_q^{X_m}(t+\tau,t)|_{\tau=0},P_0(t)\right]\neq0,
\eeq
when $\dot{f}_q(t)\neq0$.

For any $\tilde{Q}^r(t,t_{q-1})$, from Eq.~(\ref{eq:ht_evolution}) we get
\beq
\begin{split}
\frac{\partial}{\partial t}h(t,t_{q-1})=&-iV_0^\dag(t_{q-1})\tilde{Q}_r(t,t_{q-1})V_0(t_{q-1})h(t,t_{q-1})\\
=&-i\alpha_r(t)h(t,t_{q-1}).
\end{split}
\eeq
The solution of this equation is:
\beq
h(t,t_{q-1})\propto h(t_{q-1},t_{q-1})=I_{K}.
\eeq
So the horizontal lift during $[t_{q-1},t_q]$ is completely determined by $U_q(t,t_{q-1})$ up to an unimportant global phase:
\beq
V_0(t)\propto U_q(t,t_{q-1})V_0(t_{q-1}).
\eeq
At the end of this step, $V_0(t_q)\propto g_qV_0(t_{q-1})$.  From Eq.~(\ref{eq:adiabatic_evolution}), we could obtain the final state for a given initial state $\psi(0)\in\mathcal{C}$:
\beq
\begin{split}
\psi(T)\propto&V(T)V_0^\dag(0)\psi(0)\\
=&\prod_{l=1}^p g_{p-l+1} V_0(0)V_0^\dag(0)\psi(0)\\
=&\Omega_g P_0(0)\psi(0)\\
=&\Omega_g \psi(0),
\end{split}
\eeq
which is the encoded operation we desired. Note that the final Hamiltonian $H_f=\sum_{\textbf{s}}\varepsilon_{\textbf{s}}\Omega_gP_{\textbf{s}}(0)\Omega_g^\dag=H_i$, so our evolution is cyclic.
\end{proof}

\begin{remark}
Theorem.~\ref{thm:theorem1} solves the problem stated in  Sec.~\ref{sec:math_framework} to find a proper path $\lambda$ for given holonomy. Note that the requirement of a \emph{fault-tolerant circuit} in the implementation is crucial here. If it is not satisfied, the horizontal lift of $P_0(t)$ may not be completely determined by $U_q$ at each step. The condition that  $\mathscr{T}(\mathcal{G})\subset\mathcal{E}_{\text{local}}$ is not a very strong restriction. Indeed, it is always satisfied by stabilizer codes with $d\geq5$, and will generally be satisfied if we start with a fault-tolerant construction. Also note that in principle this theorem is not restricted to Clifford circuits that contain no $\pi/8$ gates. In practice, it is difficult to build the corresponding Hamiltonians constructed in our procedure, because they are hard to represent. The following theorem will show that, in order to make this process fault-tolerant, Clifford circuits are sufficient and preferred.
\end{remark}

\subsection{Fault-Tolerance of the Scheme}\label{sec:fault_tolerant}
In this section, we will discuss the fault-tolerance of the steps to realize holonomic quantum computation as presented above.
\begin{thm}\label{them:fault_tolerant}
Suppose we are given a fault-tolerant circuit $\mathcal{G}$ defined for a stabilizer code $\mathcal{C}$ with $\mathcal{E}\supseteq\mathcal{E}_{\text{local}}\supset\mathscr{T}(\mathcal{G})$. If $\mathcal{G}$ doesn't contain $\pi/8$, then by following the steps of the scheme listed in Sec.~\ref{sec:scheme}, we will get a fault-tolerant cyclic adiabatic process by the meaning of Def.~\ref{def:adiabatic_fault_tolerance}.
\end{thm}
\begin{proof}
Without loss of generality, we assume an error happens at time $t_q$ (the extension of the proof to any time $t$ is trivial). Let $P_0(t_q)=V_0(t_q)V_0^\dag(t_q)$ be the projector for the code $\mathcal{C}_{t_q}$. Since the circuit we follow is fault-tolerant, $\prod_{l=q+1}^pg_{p+q-l+1}$ maps $\mathcal{E}_{\text{local}}$ to a subset of $\mathcal{E}$, which is a correctable error set of our final code $\mathcal{C}$ (since the evolution is cyclic). It is easy to check that $\mathcal{E}_{\text{local}}$ is a correctable error set for code $\mathcal{C}_{t_q}$.  Assuming that $E^{t_q}\in \mathcal{E}_{\text{local}}$ is the error that happens at time $t_q$, it can be represented as $E^{t_q}=\sum_\mu c_\mu E^{t_q}_\mu$ where $\{E_\mu^{t_q}\}$ is a finite set of operators that spans $\mathcal{E}_{\text{local}}$. $\{E_\mu^{t_q}\}$ can always be chosen to satisfy the following error-correction condition:
\beq
P_0(t_q)E_\mu^{t_q\dag} E^{t_q}_\nu P_0(t_q) = d_{\mu\nu}P_0(t_q),
\eeq
where $d_{\mu\nu}$ is a diagonal matrix whose elements are either one or zero.  Those $E_\mu^{t_q}$ with $d_{\mu\mu}=0$ have no effect on the code space. We can always pick $K^\prime=2^{n-k}$ operators from $\{E_\mu^{t_q}\}$ with $d_{\mu\mu}=1$ to form a set $\{E_{K^\prime}^{t_q}\}$ . We can then construct another set of correctable errors with linear combination of element in $E_{K^\prime}^{t_q}$:
\beq
F_\alpha^{t_q}=\sum_{\mu=1}^{K^\prime} E_{\mu}^{t_q}R_{\mu\alpha},
\eeq
with some unitary matrix $R$ (such a unitary matrix always exists and is not unique) such that:
\beq\label{eq:eigenspace}
F_\alpha^{t_q}P_0(t)F_\alpha^{t_q\dag}=P_{\textbf{s}_\alpha}(t),\ \text{for all}\ \alpha.
\eeq
It is easy to verify that $\{F_\alpha^{t_q}\}$ still satisfies the error correction condition:
\beq\label{eq:error_correction_condition_for_F}
P_0(t_q)F_\alpha^{t_q\dag}F_\beta^{t_q}P_0(t_q)=d_{\alpha\beta}P_0(t_q).
\eeq
Now, $E^{t_q}$ can be represented by:
$E^{t_q}=\sum_\beta c^\prime_\beta F^{t_q}_\beta$.
As long as we can correct each $F_\beta^{t_q}$, we can correct $E^{t_q}$. So we consider these errors individually. If no error happens, according to the Theorem.~\ref{thm:theorem1}, the horizontal lift of $P_0(t)$ is $V_0(t)\propto U(t,0)V_0(0)$, and the final state is
\beq
\psi(T)\propto\Omega_g\psi(0),
\eeq
for $\psi(0)\in \mathcal{C}$.
The state after an error $F_\beta^{t_q}$ occurs is
\beq
\psi(t_q)\propto F_\beta^{t_q}V_0(t_q)V_0^\dag(0)\psi(0).
\eeq
According to Eq.~(\ref{eq:eigenspace}), $F_\beta^{t_q} V_0(t_q)$ represents an orthonormal frame of $P_{\textbf{s}_\beta}(t_q)$.

Next, we prove for $P_{\textbf{s}_\beta}(t)$, $t>t_q$, that the adiabatic condition Eq.~(\ref{eq:adiabatic_condition_general}) is satisfied by this scheme. We need consider only a single time segment $[t_q,t_{q+1}]$.
Again, we just treat the case of $g_{q+1}=R^x_m$ gate; the arguments for the other gates are exactly the same, since the generators of these gates are all Pauli operators. For any $\alpha\neq\beta$,
\beq
\begin{split}
&P_{\textbf{s}_\alpha}(t)\frac{\partial H(t)}{\partial t}P_{\textbf{s}_\beta}(t)=\\
&-i\frac{\pi}{4}\dot{f}_{q+1}(t)\cdot U_{q+1}^{X_m}(t,t_{q})\big(P_{\textbf{s}_\alpha}(t_q)X_m H(t_{q})P_{\textbf{s}_\beta}(t_q)\\
&-P_{\textbf{s}_\alpha}(t_q)H(t_{q})X_mP_{\textbf{s}_\beta}(t_q)\big)U^{X_m\dag}_{q+1}(t,t_{q}),
\end{split}
\eeq
where
\beq
P_{\textbf{s}_r}(t)=U_{q+1}^{X_m}(t,t_q)P_{\textbf{s}_r}(t_q)U_{q+1}^{X_m\dag}(t,t_q),
\eeq
for $r=\alpha$ or $\beta$.
Define the index set $\mathscr {I}=\{1,2,\ldots n-k\}$, and two sets $\mathscr{A}=\{j\in\mathscr{I}|[S_j(t_{q}),g_{q+1}]\neq0\}$ and
$\mathscr{B}=\mathscr{I}\backslash \mathscr{A}$. We have
\beq
\begin{split}
&P_{\textbf{s}_r}(t_{q})=\prod_{j=1}^{n-k}\frac{I+s_{r_j}S_j(t_{q})}{2}\\
=&\underbrace{\prod_{j\in\mathscr{A}}\frac{I+s_{r_j}S_j(t_{q})}{2}}_{P_{\textbf{s}_r}^{\mathscr{A}}(t_{q})}\cdot
\underbrace{\prod_{j^\prime\in\mathscr{B}}\frac{I+s_{r_{j^\prime}}S_{j^\prime}(t_{q})}{2}}_{P_{\textbf{s}_r}^{\mathscr{B}}(t_{q})}.
\end{split}
\eeq
$\mathcal{G}$ is composed of gates from $\mathscr{U}_2$, whose elements are in the normalizer of $G_n$, so at any stage $q$,  $S_j(t_{q})\in G_n$. So we have
\beq\label{eq:adiabatic_calculation_1}
\begin{split}
&P_{\textbf{s}_\alpha}(t_{q})X_mH(t_{q})P_{\textbf{s}_\beta}(t_{q})\\
=&\varepsilon_{\textbf{s}_\beta}(t_q)X_m
\prod_{j\in\mathscr{A}}\frac{I-s_{\alpha_j}S_j(t_{q})}{2}P_{\textbf{s}_\alpha}^{\mathscr{B}}(t_{q})
P_{\textbf{s}_\beta}^{\mathscr{B}}(t_{q})P_{\textbf{s}_\beta}^{\mathscr{A}}(t_{q}).
\end{split}
\eeq
Similarly, we have
\beq\label{eq:adiabatic_calculation_2}
\begin{split}
&P_{\textbf{s}_\alpha}(t_{q})H(t_{q})X_mP_{\textbf{s}_\beta}(t_{q})\\
=&\varepsilon_{\textbf{s}_\alpha}(t_q)
P_{\textbf{s}_\alpha}^{\mathscr{A}}(t_{q})P_{\textbf{s}_\alpha}^{\mathscr{B}}(t_{q})
P_{\textbf{s}_\beta}^{\mathscr{B}}(t_{q})\prod_{j\in\mathscr{A}}\frac{I-s_{\beta_j}S_j(t_{q})}{2}X_m.
\end{split}
\eeq

If $s_{\alpha_j}\neq s_{\beta_j}$ for any $j\in\mathscr{B}$, then Eq.~(\ref{eq:adiabatic_calculation_1}) and Eq.~(\ref{eq:adiabatic_calculation_2}) would be zero, and the adiabatic condition is automatically satisfied. For those $\textbf{s}_\alpha$ such that $s_{\alpha_j}=s_{\beta_j}$ for all $j\in\mathscr{B}$, we have
\beq
\begin{split}
&P_{\textbf{s}_\alpha}(t_{q})X_mH(t_{q})P_{\textbf{s}_\beta}(t_{q})\\
=&\varepsilon_{\textbf{s}_\beta}(t_q)X_m\prod_{j\in\mathscr{A}}\frac{I-s_{\alpha_j}S_j(t_{q})}{2}
\frac{I+s_{\beta_j}S_j(t_{q})}{2}P_{\textbf{s}_\beta}^{\mathscr{B}}(t_{q}),
\end{split}
\eeq
and
\beq
\begin{split}
&P_{\textbf{s}_\alpha}(t_{q})H(t_{q})X_mP_{\textbf{s}_\beta}(t_{q})\\
=&\varepsilon_{\textbf{s}_\alpha}(t_q)P_{\textbf{s}_\alpha}^{\mathscr{B}}(t_{q})\prod_{j\in\mathscr{A}}\frac{I+s_{\alpha_j}S_j(t_{q})}{2}
\frac{I-s_{\beta_j}S_j(t_{q})}{2}X_m.
\end{split}
\eeq
The above two expressions are nonzero only if $s_{\beta_j}=-s_{\alpha_j}$ for all $j\in\mathscr{A}$. Therefore, there is only one $\beta$ such that $P_{\textbf{s}_\alpha}\dot{H}(t)P_{\textbf{s}_\beta}\neq0$ and hence needs further calculation. For that specific $\textbf{s}_\beta$, we have a simple relation:
\beq
X_mP_{\textbf{s}_\alpha}(t_{q})X_m=P_{\textbf{s}_\beta}(t_q).
\eeq
We obtain
\beq
\begin{split}
&\big\|P_{\textbf{s}_\alpha}(t)\dot{H}(t)P_{\textbf{s}_\beta}(t)\big\|_1\\
=&\frac{\pi}{4}\dot{f}_{q+1}(t)\big\|\varepsilon_{\textbf{s}_\beta}(t_q)
X_m P_{\textbf{s}_\beta}(t_q)-\varepsilon_{\textbf{s}_\alpha}(t_q)P_{\textbf{s}_\alpha}(t_q)
X_m\big\|_1\\
=&\frac{\pi}{4}\dot{f}_{q+1}(t)K\cdot\big|\varepsilon_{\textbf{s}_\alpha}(t_q)-
\varepsilon_{\textbf{s}_\beta}(t_q)\big|.
\end{split}
\eeq
So the LHS of Eq.~(\ref{eq:adiabatic_condition_general}) reduces to
\beq\label{eq:error_estimation} \frac{\pi\dot{f}_{q+1}(t)}{4\big|\varepsilon_{\textbf{s}_\alpha}(t_q)-
\varepsilon_{\textbf{s}_\beta}(t_q)\big|}.
\eeq
If we are in Step 3, since $|\mathscr{A}|$ is odd, we have
\beq
|\varepsilon_{\textbf{s}_\alpha}(t_q)-\varepsilon_{\textbf{s}_{\beta}}(t_q)|
=\big|-\sum_{j\in\mathscr{A}}2s_{\alpha_j}\big|\geq2,
\eeq
and if we are in Step 4, because of our operation to break the degeneracy by setting $C_b=0.5$, we have
\beq
|\varepsilon_{\textbf{s}_\alpha}(t_q)-\varepsilon_{\textbf{s}_{\beta}}(t_q)|
=\Big|-\sum_{\substack{j\in\mathscr{A}\\j\neq b}}2s_{\alpha_j}-s_{\alpha_{b}}\Big|\geq1.
\eeq
If $\frac{\pi}{4}\dot{f}_{q+1}(t)\ll 1$ is satisfied, which is always possible, then $P_{\textbf{s}_\beta}(t)$ satisfies the adiabatic condition for time segment $[t_q,t_{q+1}]$. The same argument can be applied to the rest of the time segments to show that the adiabatic condition can always be satisfied by choosing appropriate functions $f(t)$.

Now, we can use the evolution equation  Eq.~(\ref{eq:adiabatic_evolution}):
\beq\label{eq:error_evolution_equation}
\begin{split}
\psi_\beta(T)=& V_{\textbf{s}_\beta}(T)V_0^\dag(t_q)F_\beta^{t_q\dag}F_\beta^{t_q} V_0(t_q)V_0^\dag(0)\psi(0)\\
=&V_{\textbf{s}_\beta}(T)\underbrace{V_0^\dag(t_q)V_0(t_q)}_{I_K}V_0^\dag(t_q)F_\beta^{t_q\dag}F_{\beta}^{t_q}\\
&\times V_0(t_q)\underbrace{V_0^\dag(t_q)V_0(t_q)}_{I_K}V_0^\dag(0)\psi(0)\\
=&V_{\textbf{s}_\beta}(T)V_0^\dag(t_q)d_{\beta\beta}V_0(t_q)V_0^\dag(t_q)V_0(t_q)V_0^\dag(0)\psi(0)\\
=&V_{\textbf{s}_\beta}(T)V_0^\dag(0)\psi(0),\\
\end{split}.
\eeq
where the third equality follows from Eq.~(\ref{eq:error_correction_condition_for_F}). Here, $V_{\textbf{s}_\beta}(t)$ is defined as
\beq
V_{\textbf{s}_\beta}(t)=U(t,t_q)F_\beta^{t_q}V_0(t_q)h(t,t_q) \ \text{for }t>t_q,
\eeq
as the horizontal lift of $P_{\textbf{s}_\beta}(t)$ given initial condition $V_{\textbf{s}_\beta}(t_q)=F_\beta^{t_q}V_0(t_q)$.

Again, we just focus on one time segment $t\in[t_q,t_{q+1}]$ and the single gate $R^x_m$, since the rest are just the same:
\beq
\begin{split}
\tilde{Q}^{R_m^x}(t,t_q)=&P_{\textbf{s}_\beta}(t_q)Q^{R^x_m}(t,t_q)P_{\textbf{s}_\beta}(t_q)\\
=&-\frac{\pi}{4}\dot{f}_q(t)P_{\textbf{s}_\beta}(t_q)X_mP_{\textbf{s}_\beta}(t_q)=0.
\end{split}
\eeq
So $h(t,t_q)=I_K$, and according to Lemma~\ref{lemma:horiziontal_uniqueness},
\beq
V_{\textbf{s}_\beta}(t)=U(t,t_q)F_\beta^{t_q}V_0(t_q)
\eeq
is the only horizontal lift of $P_{\textbf{s}_\beta}(t)$ for $t>t_q$. Following Eq.~(\ref{eq:error_evolution_equation}), we have
\beq
\begin{split}
\psi_\beta(T)=& U(T,t_q)F_\beta^{t_q}V_0(t_q)V_0^\dag(0)\psi(0)\\
=&F_\beta^T U(T,t_q)V_0(t_q)V_0^\dag(0)\psi(0)\\
=&F_\beta^T U(T,0)V_0(0)V_0^\dag(0)\psi(0)\\
=&F_\beta^T\Omega_g\psi(0),
\end{split}
\eeq
where $F_\beta^T$ is defined as $U(T,t_q)F_\beta^{t_q} U^{\dag}(T,t_q)$, with $U(T,t_q)=\prod_{l=q+1}^p g_{p+q-l+1}$.
Since the circuit we follow is fault-tolerant, $F_\beta^T\in \mathcal{E}$.
Taking dynamic phases into account, if $E^{t_q}$ occurs, the final state should be
\beq
\psi(T)=\sum_\beta c_\beta^\prime \exp{\left(-i\int_{t_1}^T\varepsilon_{\textbf{s}_\beta}(t)\text{d}t\right)}
F_\beta^T\Omega_g\psi(0),
\eeq
so the adiabatic process we propose here is fault-tolerant by the meaning of Def.~\ref{def:adiabatic_fault_tolerance}.

\end{proof}
\begin{remark}
We now discuss some details of the adiabatic theorem and its application to our scheme. The traditional version of the adiabatic theorem stated in ~\cite{Messiah:1965:North} guarantees that the adiabatic approximation is satisfied with precision $\delta\leq\epsilon^2$ during the adiabatic evolution, if the condition
\beq
\frac{\sup_{t\in[0,T]}\parallel P_{\textbf{s}_\alpha}(t)\frac{\partial}{\partial t}H(t)P_{\textbf{s}_\beta}(t)\parallel_1}{\inf_{t\in[0,T]}K\left(\varepsilon_{\textbf{s}_\alpha}(t)
-\varepsilon_{\textbf{s}_\beta}(t)\right)^2}\leq \epsilon, \ \text{for any}\ \alpha\neq\beta,
\eeq
is satisfied (note that $K=2^k$ is the dimension of the code space), which is equivalent to:
\beq
\sup_{q,t\in[0,T]}\frac{\pi\dot{f}_{q}(t)}{4}\leq\epsilon
\eeq
in our case.
However, it is known that this statement is neither sufficient nor necessary, and under certain conditions on the Hamiltonian, we can obtain better results~\cite{Hagedorn:2002:235,lidarAdiabaticaccuracy:102106}.
According to Ref.~\cite{lidarAdiabaticaccuracy:102106}, for a Hamiltonian $H(\vartheta)$($\vartheta=t/T$) that is analytic near $[0,1]$ in the complex plane, with the absolute value of the imaginary part of the nearest pole being $\gamma$, and the first $\mathcal{N}\geq1$ derivatives at boundaries equal to zero, i.e., $H^{(l)}(0)=H^{(l)}(1)=0$ for $l\leq\mathcal{N}$, if we set
\beq\label{eq:evolution_time}
T=\frac{e}{\gamma}\mathcal{N}\frac{\xi^2}{d_{\text{min}}^3},
\eeq
where
\beq
\xi=\sup_{\vartheta\in[0,1]}\parallel\text{d}H/\text{d}\vartheta\parallel_\infty,
\eeq
($\parallel\cdot\parallel_\infty$ is standard operator norm, and $d_{\text{min}}$ is the minimum spectral gap) then the adiabatic approximation error satisfies
\beq
\delta_{\text{ad}}\leq(\mathcal{N}+1)^{\gamma+1}e^{-\mathcal{N}},
\eeq
or equivalently,
\beq
\delta_{\text{ad}}\lesssim (cT+1)^{\gamma+1}e^{-cT},
\eeq
with $c=\frac{\gamma d^3}{e\xi^2}$. This means that we can decrease the adiabatic error exponentially with evolution time $T$, which is  proportional to $\mathcal{N}$.
Applying this theorem to our piecewise adiabatic evolution, for the $q$th segment, we set $T_q=t_{q}-t_{q-1}$ to be $\frac{e}{\gamma}\mathcal{N}\xi_q^2$, where $\xi_q$ is $\xi$ defined on the $q$th time segment, and $f_q(t)$ is chosen such that a) the boundary condition mentioned above is satisfied, and b) $H(\vartheta)$ is analytic near $[0,1]$. The adiabatic error for an encoded unitary operation composed of $p$ gates can therefore be bounded by
\beq
\delta_{\text{ad}}\lesssim p\cdot\sup_{1\leq q\leq p}(c_qT_q+1)^{\gamma+1}e^{-c_qT_q}.
\eeq

During the adiabatic process, the energy gap between the ground space and any other eigenspace is lower-bounded by 1, and this does not decrease with the size of the problem or the number of levels of code concatenation. Again, we assume the qubits are independently coupled to the thermal environment and the corresponding thermal errors are local and low-weight during certain period of evolution time. Those low-weight thermal excitations will cause transitions from the ground space to higher energy excited spaces. Their rate will decrease exponentially with the existence of an energy gap ~\cite{TameemNJP:adiabaticMarkovianME},
$\delta_{\text{thermal}}\sim O\left(\exp(-d_{\min})\right)$,
while Eq.~(\ref{eq:evolution_time}) shows that the time needed to finish the process grows inversely as the cube of the minimum gap. The system error can be bounded by the sum of these two errors~\cite{LidarTowardFTAdqcPhysRevLett.100.160506}:
\beq
\delta_S<\delta_{\text{thermal}}+\delta_{\text{ad}}.
\eeq
So, qualitatively, we can make both the adiabatic error and thermal excitation exponentially small with efficient overhead in processing time. (However see the discussion in Sec.~\ref{sec:conclusion} for possible limitations of this argument.) 
\end{remark}

Theorem~\ref{them:fault_tolerant}  builds an equivalence relation between a fault-tolerant encoded Clifford circuit and a fault-tolerant adiabatic process that gives the same encoded unitary operation. For most stabilizer codes (e.g., Steane code, the surface code, or the C4 code), encoded operations in the Clifford group can all be realized by such fault-tolerant circuits. Standard techniques, like magic state injection and distillation, can realize fault-tolerant encoded non-Clifford gates like the encoded $\pi/8$ and encoded Toffoli gates. Magic state distillation can be implemented by fault-tolerant encoded unitary gates from the Clifford group. Thus, this holonomic scheme is universal for fault-tolerant quantum computation.

\section{Examples}\label{sec:examples}
In this section, we apply the procedure developed above to construct adiabatic processes for specific codes. For pedagogical purposes, our first example realizes the encoded $X$ for the simplest 3-qubit repetition code. Our second example is the 7 qubit Steane code. This example is of practical importance because through concatenation of this simple code, fault-tolerant quantum computation can be achieved when the error rate is lower than the threshold.
\subsection{3-qubit repetition code}\label{sec:3qubit}
There are two stabilizer generators for this code, as shown below:
\begin{center}
\begin{tabular}{ c|c c c  }
  \hline
  \hline
  $   $ & 1 & 2 & 3  \\
  \hline
  $S_1$ & $Z$ & $Z$ & $I$ \\
  $S_2$ & $I$ & $Z$ & $Z$ \\
  \hline
  \hline
\end{tabular}
\end{center}
The encoded $X$ operator for this code is $\overline{X}=X_1 X_2 X_3$. The encoded $X$ gate can be performed by a circuit like $\overline{X}=R^x_1 R^x_1 R^x_2 R^x_2 R^x_3 R^x_3$, so the process takes 6 steps. The initial Hamiltonian is
\beq
H(0)=-Z_1Z_2-Z_2Z_3.
\eeq
For the first step of the adiabatic process, we have $[X_1, Z_1Z_2]\neq0$ for $t\in[0,t_1]$. So the Hamiltonian during that interval is
\beq
H(t)=-\cos\left(f_1(t)\frac{\pi}{2}\right)Z_1Z_2
+\sin\left(f_1(t)\frac{\pi}{2}\right)Y_1Z_2-Z_2Z_3,
\eeq
with $H(t_1)=Y_1Z_2-Z_2Z_3$. In the second step, for $t\in[t_1,t_2]$, the Hamiltonian is:
\beq
H(t)=\cos\left(f_2(t)\frac{\pi}{2}\right)Y_1Z_2
+\sin\left(f_2(t)\frac{\pi}{2}\right)Z_1Z_2-Z_2Z_3,
\eeq
with $H(t_2)=Z_1Z_2-Z_2Z_3$. In the third step, we see that $[X_2,Z_1Z_2]\neq 0$ and $[X_2,Z_2Z_3]\neq 0$, which implies that there might exist a $P_{\textbf{s}}$ that doesn't satisfy the adiabatic condition. Actually, as shown on the left side of Fig.~\ref{Fig:energy_diagram},
\begin{figure}[!ht]
\includegraphics[width=80mm]{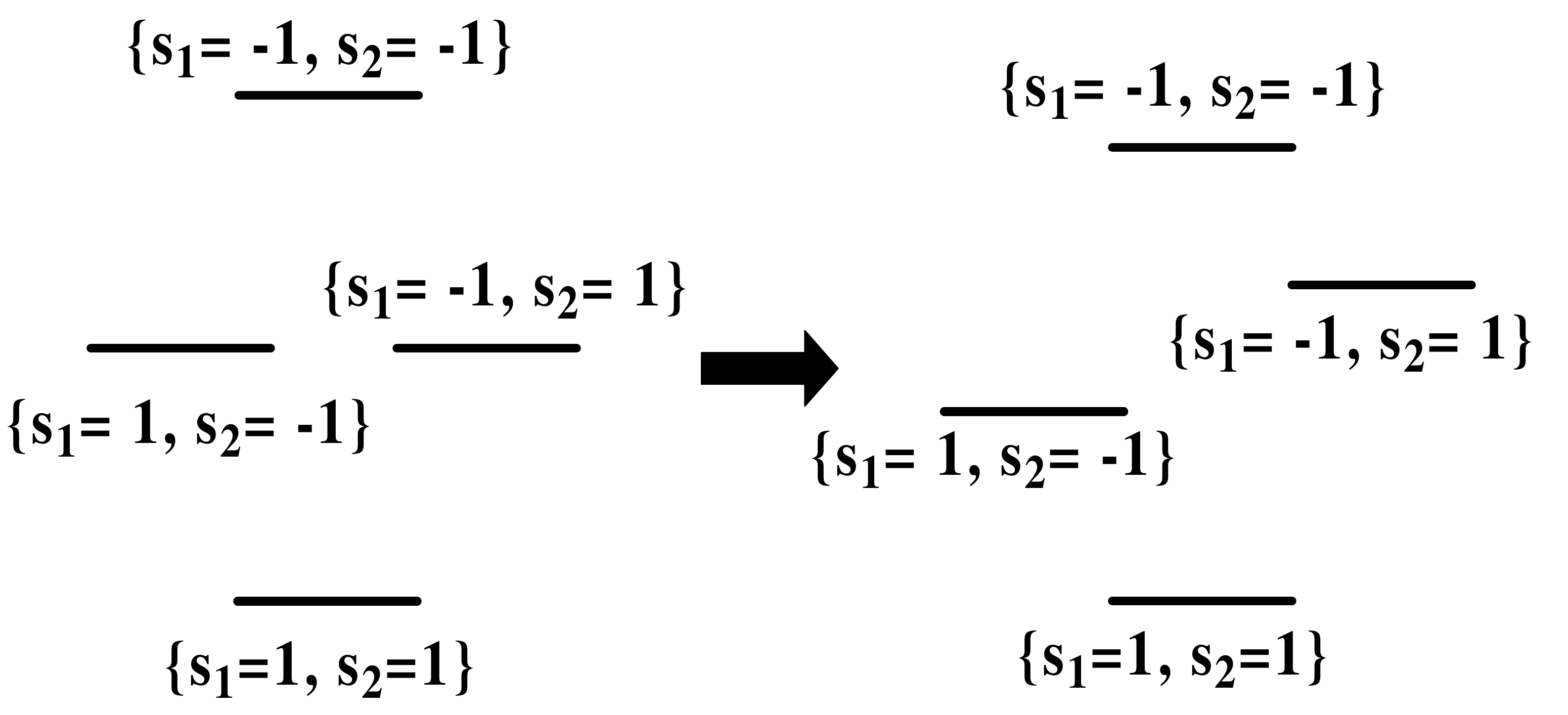}
\caption{\label{Fig:energy_diagram} The variation of the energy diagram at the beginning of the third and fourth step to break the degeneracy of space $\{s_1=1,s_2=-1\}$ and $\{s_1=-1,s_2=1\}$ .}
\end{figure}
we have $[P_{\textbf{s}=\{1,-1\}}(t_2)+P_{\textbf{s}=\{-1.1\}}(t_2),X_2]=0$, which means both $P_{\textbf{s}=\{1,-1\}}(t_2)$ and $P_{\textbf{s}=\{-1.1\}}(t_2)$ will not satisfy the adiabatic condition during the evolution if we do not break this degeneracy. Following the scheme in Sec.~\ref{sec:scheme}, we alter the Hamiltonian at $t_2$ to be
$H(t_2)=Z_1Z_2-0.5Z_2Z_3$. The corresponding energy diagram is shown on the right side of Fig.~\ref{Fig:energy_diagram}. Then we vary the Hamiltonian in the following way for $t\in[t_2,t_3]$:
\beq
\begin{split}
H(t)=& \cos\left(f_3(t)\frac{\pi}{2}\right)Z_1Z_2 - \sin\left(f_3(t)\frac{\pi}{2}\right) Z_1Y_2\\
&-0.5 \cos\left(f_3(t)\frac{\pi}{2}\right)Z_2Z_3+0.5\sin\left(f_3(t)\frac{\pi}{2}\right) Y_2Z_3,
\end{split}
\eeq
with $H(t_3)=-Z_1Y_2+Y_2Z_3$. We continue to break the degeneracy in the fourth step, since $X_2$ again does not commute with both $P_{\textbf{s}=\{1,-1\}}(t_3)$ and $P_{\textbf{s}=\{1,-1\}}(t_3)$. So for $t\in[t_3,t_4]$, the Hamiltonian is:
\beq
\begin{split}
H(t)=&-\cos\left(f_4(t)\frac{\pi}{2}\right)Z_1Y_2-\sin\left(f_4(t)\frac{\pi}{2}\right)Z_1Z_2\\
&+0.5\cos\left(f_4(t)\frac{\pi}{2}\right)Y_2Z_3+0.5\sin\left(f_4(t)\frac{\pi}{2}\right)Z_2Z_3,
\end{split}
\eeq
with $H(t_4)=-Z_1Z_2+0.5Z_2Z_3$, which can then be restored to $H(t_4)=-Z_1Z_2+Z_2Z_3$. The fifth and sixth steps are just like the first and the second steps. The final Hamiltonian is $H(T=t_6)=-Z_1Z_2-Z_2Z_3$, which is equal to the initial Hamiltonian, and we obtain our geometric encoded $X$ operation.
\begin{remark}
The encoded $Z$ operator for this code is $\overline{Z}=Z_1Z_2Z_3$. However, we cannot use our scheme to build the adiabatic process according to this circuit, since $Z_1$, $Z_2$ and $Z_3$ all commute with the initial Hamiltonian $H(0)$. We can see that for this simple code, there doesn't exist an $\mathcal{E}$ that includes $\{Z_1,Z_2,Z_3\}$, so the conditions for both Theorem~\ref{thm:theorem1} and Theorem~\ref{them:fault_tolerant} are not satisfied.
\end{remark}

\subsection{The Steane code}\label{sec:Steane's_code}
Fault-tolerant quantum computation can be realized through concatenation of Steane code. In our scheme, we only apply our scheme \emph{at the bottom level of concatenation}. For higher levels, encoded unitary operations and error correction are done in the usual way. By doing so, we keep the constant energy gap between the ground space and error spaces and thus maintain the ability to suppress low weight thermal errors, and we can bound the weight of terms in the system Hamiltonian. One set of  generators of the stabilizer group of the Steane code is listed below:
\begin{center}
\begin{tabular}{ c|c c c c c c c }
  \hline
  \hline
  $   $ & 1 & 2 & 3 & 4 & 5 & 6 & 7  \\
  \hline
  $S_1$ & $X$ & $I$ & $X$ & $X$ & $I$  & $X$ & $I$\\
  $S_2$ & $I$ & $X$ & $X$ & $I$ & $I$  & $X$ & $X$ \\
  $S_3$ & $X$ & $I$ & $X$ & $I$ & $X$  & $I$ & $X$ \\
  $S_4$ & $Z$ & $I$ & $Z$ & $Z$ & $I$  & $Z$ & $I$ \\
  $S_5$ & $I$ & $Z$ & $Z$ & $I$ & $I$  & $Z$ & $Z$ \\
  $S_6$ & $Z$ & $I$ & $Z$ & $I$ & $Z$  & $I$ & $Z$  \\
  \hline
  \hline
\end{tabular}
\end{center}
The circuits for the encoded Hadmard, encoded $S$ and encoded $X$ and $Z$ gates are all bit-wise transversal, and thus naturally fault-tolerant. Note that each Hadmard can be decomposed into $SR^xS$ up to a global phase. The geometric realizations of such gates are similar to that of the encoded $X$ for the 3-qubit repetition code shown above. So in this section we focus on the CNOT gate.

\subsubsection{CNOT Gate}\label{sec:CNOT}
For the Steane code, the fault-tolerant encoded ${\rm CNOT}_{1\rightarrow2}$ (control qubit encoded in block 1, target qubit encoded in block 2) can be realized transversally between two blocks of qubits. We illustrate our scheme for one pair of qubits from the two blocks, all the other operations are the same. The initial Hamiltonian $H(0)$ can be written as $-\sum_{j=1}^7\left(S_j^1+S_j^2\right)$, where $S^i_{j}$ is the $j$th generator for the $i$th block.
Each physical CNOT can be decomposed into
\vspace{1.5em}
\centerline{
\Qcircuit @C=1.4em @R=1.0em {
&\ctrl{1}& \qw &\ccteq{1} \\
&\targ &\qw &\ccteqg  \\
}
}
\vspace{1.5em}
\centerline{
\Qcircuit @C=1.1em @R=1.0em {
&\qw &\qw &\qw &\qw&\multigate{1}{R^{zz}} & \gate{S}&\qw &\qw &\qw \\
&\gate{S}&\gate{R^x}&\gate{S}&\gate{S}& \ghost{R^{zz}} &\gate{S}&\gate{R^x}&\gate{S}&\qw \\
}
}
up to a global phase.

As we can see, there are nine gates in the circuit. The transformation of the first and last four single qubit gates has been discussed before. We will just show the Hamiltonian during time interval $[t_4,t_5]$ when the two-qubit gate is performed.

The Hamiltonian at time $t_4$ can be shown to be:
\beq
\begin{split}
H(t_4)=&-\sum_{j=1}^7S_j^1-S_2^2-S_5^2
-Z^2_1X^2_3X^2_4X^2_6-Z^2_1X^2_3X^2_5X^2_7\\
&-Y^2_1Z^2_3Z^2_4Z^2_6-Y^2_1Z^2_3Z^2_5Z^2_7.
\end{split}
\eeq
For an $R^{zz}$ gate acting on qubit 1 in both blocks, we have
$[Z^1_1Z^2_1,X^1_1X^1_3X^1_4X^1_6]\neq0$, $[Z^1_1Z^2_1,X^1_1X^1_3X^1_5X^1_7]\neq0$, $[Z^1_1Z^2_1,Y^2_1Z^2_3Z^2_4Z^2_6]\neq0$ and $[Z^1_1Z^2_1,Y^2_1Z^2_3Z^2_5Z^2_7]\neq0$ for $t\in[t_4,t_5]$, so the Hamiltonian during this interval can be chosen to be:
\beq
\begin{split}
&H(t)=\\
&-\sum_{\substack{j\neq1\\j\neq3}}S_j^1-S_2^2-S_5^2-Z^2_1X^2_3X^2_4X^2_6-Z^2_1X^2_3X^2_5X^2_7\\
&-\cos\left(f_5(t)\frac{\pi}{2}\right)X^1_1X^1_3X^1_4X^1_6+\sin\left(f_5(t)\frac{\pi}{2}\right)Y^1_1X^1_3X^1_4X^1_6Z^2_1\\
&-\cos\left(f_5(t)\frac{\pi}{2}\right)X^1_1X^1_3X^1_5X^1_7+\sin\left(f_5(t)\frac{\pi}{2}\right)Y^1_1X^1_3X^1_5X^1_7Z^2_1\\
&-\cos\left(f_5(t)\frac{\pi}{2}\right)Y^2_1Z^2_3Z^2_4Z^2_6-\sin\left(f_5(t)\frac{\pi}{2}\right)Z^1_1X^2_1Z^2_3Z^2_4Z^2_6\\
&-0.5\cos\left(f_5(t)\frac{\pi}{2}\right)Y^2_1Z^2_3Z^2_5Z^2_7\\
&-0.5\sin\left(f_5(t)\frac{\pi}{2}\right)Z^1_1X^2_1Z^2_3Z^2_5Z^2_7,
\end{split}
\eeq
with
\beq
\begin{split}
H(t_5)=&-\sum_{\substack{j\neq1\\j\neq3}}S_j^1-S_2^2-S_5^2-Z^2_1X^2_3X^2_4X^2_6-Z^2_1X^2_3X^2_5X^2_7\\
&+Y^1_1X^1_3X^1_4X^1_6Z^2_1+Y^1_1X^1_3X^1_5X^1_7Z^2_1\\
&-Z^1_1X^2_1Z^2_3Z^2_4Z^2_6-Z^1_1X^2_1Z^2_3Z^2_5Z^2_7.
\end{split}
\eeq

After all nine gates have been performed, the Hamiltonian will be:
\beq\label{eq:CNOT_first_part_qubit}
\begin{split}
H(T_1)=&-\sum_{\substack{j\neq1\\j\neq3}}(S_j^1+S_j^2)\\
&-X_1^1X_3^1X_4^1X_6^1X_1^2-X_1^1X_3^1X_5^1X_7^1X_1^2\\
&-Z_1^1Z_1^2Z_3^2Z_4^2Z_6^2-Z_1^1Z_1^2Z_3^2Z_5^2Z_7^2.
\end{split}
\eeq

After repeating this procedure on all 7 pairs of qubits, the final Hamiltonian will be:
\beq\label{eq:CNOT_final_Hamiltonian}
\begin{split}
H(T)=&-X_1^1X_3^1X_4^1X_6^1X_1^2X_3^2X_4^2X_6^2\\
     &-X_2^1X_3^1X_6^1X_7^1X_2^2X_3^2X_6^2X_7^2\\
     &-X_1^1X_3^1X_5^1X_7^1X_1^2X_3^2X_5^2X_7^2\\
     &-Z_1^1Z_3^1Z_4^1Z_6^1-Z_2^1Z_3^1Z_6^1Z_7^1-Z_1^1Z_3^1Z_5^1Z_7^1\\
     &-X_1^2X_3^2X_4^2X_6^2-X_2^2X_3^2X_6^2X_7^2-X_1^2X_3^2X_5^2X_7^2\\
     &-Z_1^1Z_3^1Z_4^1Z_6^1Z_1^2Z_3^2Z_4^2Z_6^2\\
     &-Z_2^1Z_3^1Z_6^1Z_7^1Z_2^2Z_3^2Z_6^2Z_7^2\\
     &-Z_1^1Z_3^1Z_5^1Z_7^1Z_1^2Z_3^2Z_5^2Z_7^2,
\end{split}
\eeq
which is equal to $H(0)$. Although the final Hamiltonian equals to the initial Hamiltonian, the maximum weight of its elements has doubled, which is not good for practical implementation. Although recent results have shown how to map such Hamiltonians to more physically reasonable two-body interactions \cite{RydbergsimulatorNatPhysics, KempeQuantumGadget, Oliveira:2008:CQS:2016985.2016987}, it is still important to decrease the maximum weight of the Hamiltonian terms.

\subsubsection{Lowering the weight of the Hamiltonian}
The maximum weight of the terms in the final Hamiltonian is 8, compared to the 4 for the initial Hamiltonian given by Eq.~(\ref{eq:CNOT_final_Hamiltonian}). This problem may not exist for a fault-tolerant scheme based on the surface code~\cite{Folwer2012PhysRevA.86.032324}, since during the process of code deformation, the weight of the stabilizer generators is always bounded by 4. However, even for the Steane code, we can lower the weight of the Hamiltonian terms during the process by using the following observation:
\beq
H(t)=-\sum_jS_j(t)=U_LH(t)U_L^\dag=-\sum_jU_LS_j(t)U_L^\dag,
\eeq
for some unitary operator $U_L$ which commutes with $H(t)$, but not necessarily with the individual terms $S_j(t)$. This means that the decomposition of the Hamiltonian is not unique, and we can take advantage of this freedom. Here we set $U_L=\overline{\rm CNOT}_{12}$. After the first three transversal CNOTs, the Hamiltonian would become:
\beq
\begin{split}
H(T_3)=&-X_1^1X_3^1X_4^1X_6^1X_1^2X_3^2-X_2^1X_3^1X_6^1X_7^1X_2^2X_3^2\\
     &-X_1^1X_3^1X_5^1X_7^1X_1^2X_3^2-Z_1^1Z_3^1Z_4^1Z_6^1\\
     &-Z_2^1Z_3^1Z_6^1Z_7^1-Z_1^1Z_3^1Z_5^1Z_7^1\\
     &-X_1^2X_3^2X_4^2X_6^2-X_2^2X_3^2X_6^2X_7^2\\
     &-X_1^2X_3^2X_5^2X_7^2-Z_1^1Z_3^1Z_1^2Z_3^2Z_4^2Z_6^2\\
     &-Z_2^1Z_3^1Z_2^2Z_3^2Z_6^2Z_7^2-Z_1^1Z_3^1Z_1^2Z_3^2Z_5^2Z_7^2.
\end{split}
\eeq
The maximum weight of any term is 6. This Hamiltonian is equal to the following form:
\beq
\begin{split}
H(T_3)=&-X_1^1X_3^1X_4^1X_6^1X_4^2X_6^2-X_2^1X_3^1X_6^1X_7^1X_6^2X_7^2\\
     &-X_1^1X_3^1X_5^1X_7^1X_5^2X_7^2-Z_1^1Z_3^1Z_4^1Z_6^1\\
     &-Z_2^1Z_3^1Z_6^1Z_7^1-Z_1^1Z_3^1Z_5^1Z_7^1\\
     &-X_1^2X_3^2X_4^2X_6^2-X_2^2X_3^2X_6^2X_7^2\\
     &-X_1^2X_3^2X_5^2X_7^2-Z_4^1Z_6^1Z_1^2Z_3^2Z_4^2Z_6^2\\
     &-Z_6^1Z_7^1Z_2^2Z_3^2Z_6^2Z_7^2-Z_5^1Z_7^1Z_1^2Z_3^2Z_5^2Z_7^2,
\end{split}
\eeq
which again has maximum weight of 6. Note that the transition between these two Hamiltonian representations can be arbitrarily fast, since they are equal.
Then, if we perform the remaining four transversal CNOTs by an adiabatic process, the final Hamiltonian will return to the initial one represented by $H(T)=-\sum_{j=1}^7(S_j^1+S_j^2)$. During the whole process, the maximum weight of terms in the Hamiltonian is reduced from 8 to 6. The weight does not increase when the codes are concatenated, since we just apply our scheme to the bottom level of concatenation and do higher levels operations in their usual way. So we can realize HQC fault-tolerantly with maximum Hamiltonian weight 6 while keeping the constant energy gap. For a small code like the Steane code, this may be the best we can do. For larger block codes, especially for topological codes like the surface code, it is very likely that the weight of the Hamiltonian terms during the adiabatic process can be well bounded.

\section{Summary and Conclusion}\label{sec:conclusion}
We have described a scheme for fault-tolerant HQC on a stabilizer code, which takes advantage of a constant energy gap during the process as well as of the intrinsic resilience of HQC. We've shown that from a fault-tolerant circuit without $\pi/8$ gates, we can systematically construct a fault-tolerant adiabatic process that implements the very same encoded unitary operation as the original circuit, with information encoded in the ground state of the system Hamiltonian. As long as we can realize holonomic versions of gates in the Clifford group, we can implement fault-tolerant universal quantum computation by using magic state distillation.

Holonomic single-qubit operation has been recently realized on trapped single $^{40}\text{Ca}^+$ ion system through adiabatic evolution~\cite{HQC_adiabatic_realization2013}. Theoretical work on non-adiabatic non-abelian HQC has also been proposed~\cite{NJPNonAdiabaticHolonomic}, and corresponding experiments have recently been realized in superconducting qubits~\cite{NonAbelianNonAdiabaticHolonomicNature} and NMR~\cite{FengGuanruPhysRevLett.110.190501}. Applying our strategy to actual physical systems will need certain techniques, like quantum gadgets~\cite{KempeQuantumGadget, Oliveira:2008:CQS:2016985.2016987} or the digital quantum simulator~\cite{RydbergsimulatorNatPhysics}, to build the many-body interactions. If the system Hamiltonian is built in one of these effective ways, rather than being fundamental, it may dramatically change the local error model we have assumed. This effect needs further investigation.

Our fault-tolerant HQC scheme differs from  adiabatic gate teleportation (AGT) \cite{Bacon-AGTPhysRevLett.103.120504}, and the scheme in Refs.~\cite{OgyanHolonomicQCPhysRevLett.102.070502,OgyanHolonomicQcPhysRevA.80.022325}, in the following ways: 1. Instead of focusing on single qubit unitary operations or two-qubit unitary operations, our scheme obtains holonomy through directly dragging the ground space (code space) of the system, whose path in the Grassmann manifold forms a closed loop. 2. During the adiabatic procedure, the energy spectrum of system basically remains the same, and there always exists an energy gap between the code space and the excited spaces, so that information is protected from low weight thermal excitation by an energy gap.

There are several advantages over other schemes here. We can reduce the low-weight error rate at the bottom level of code concatenation due to the existence of an energy gap, so that the frequency of error correction procedure can be greatly reduced at bottom physical level. The measurement of stabilizer generators and subsequent error correction can themselves introduce more errors, and this is one of the reasons why thresholds for current fault-tolerant schemes are so low. Moreover, the number of physical qubits needed in a fault-tolerant scheme is strongly dependent on the error rate at the physical level. Error rates substantially smaller than the threshold allow much smaller numbers of physical qubits. Eventually, our scheme may reduce the resource overhead needed to do fault-tolerant quantum computation.

On the other hand, our scheme seems naturally compatible with Hamiltonian-protected quantum memories~\cite{Wotton:2012QMemoryReview}, and has the potential to do fault-tolerant computation based on those kind of memories, especially those with self-correcting ability \cite{chesi2010self-correciton_memory, hutter2012self-correction_memory}.
No dynamical method seems capable of manipulating the topological degrees of freedom encoded in the ground space of these memory in the presence of a system Hamiltonian, as far as we know, since they would introduce terms that do not commute with the system Hamiltonian. However, our method of locally deforming these Hamiltonians could potentially do quantum computation on such systems. We conjecture that during this kind of local deformation procedure, the system will keep its self-correction capability in a thermal environment while quantum computation is implemented. This interesting topic requires further investigation, and may open a new way of studying quantum computater architecture.

We note that this method to construct fault-tolerant HQC is basically a serial procedure from gate to gate. For circuits with large depth, we could investigate the possibility of parallel operation. For large block codes or topological codes, such parallelization can be done, and is crucial in practice. We hope to apply our method to fault-tolerant schemes based on large block codes and topological codes, which may have higher thresholds than fault-tolerant schemes that concatenate small codes. Very likely, the maximum weight of the Hamiltonian terms used to describe topological codes during adiabatic evolution will be small and well bounded.

\section{ACKNOWLEGEMENTS}
We would like to  thank Daniel Lidar for discussion about the adiabatic theorem. Y.-C.Z \& T.A.B acknowledge support from NSF Grants No. EMT-0829870 and No. TF-0830801, and from the ARO MURI Grant W911NF-11-1-0268.


\begin{thebibliography}{10}%
\makeatletter
\providecommand \@ifxundefined [1]{%
 \ifx #1\undefined \expandafter \@firstoftwo
 \else \expandafter \@secondoftwo
\fi
}%
\providecommand \@ifnum [1]{%
 \ifnum #1\expandafter \@firstoftwo
 \else \expandafter \@secondoftwo
\fi
}%
\providecommand \enquote [1]{``#1''}%
\providecommand \bibnamefont  [1]{#1}%
\providecommand \bibfnamefont [1]{#1}%
\providecommand \citenamefont [1]{#1}%
\providecommand\href[0]{\@sanitize\@href}%
\providecommand\@href[1]{\endgroup\@@startlink{#1}\endgroup\@@href}%
\providecommand\@@href[1]{#1\@@endlink}%
\providecommand \@sanitize [0]{\begingroup\catcode`\&12\catcode`\#12\relax}%
\@ifxundefined \pdfoutput {\@firstoftwo}{%
 \@ifnum{\z@=\pdfoutput}{\@firstoftwo}{\@secondoftwo}%
}{%
 \providecommand\@@startlink[1]{\leavevmode}%
 \providecommand\@@endlink[0]{}%
}{%
 \providecommand\@@startlink[1]{%
  \leavevmode
  \pdfstartlink
   attr{/Border[0 0 1 ]/H/I/C[0 1 1]}%
   user{/Subtype/Link/A<</Type/Action/S/URI/URI(#1)>>}%
  \relax
 }%
 \providecommand\@@endlink[0]{\pdfendlink}%
}%
\providecommand \url  [0]{\begingroup\@sanitize \@url }%
\providecommand \@url [1]{\endgroup\@href {#1}{\urlprefix}}%
\providecommand \urlprefix [0]{URL }%
\providecommand \Eprint[0]{\href }%
\@ifxundefined \urlstyle {%
  \providecommand \doi [1]{doi:\discretionary{}{}{}#1}%
}{%
  \providecommand \doi [0]{doi:\discretionary{}{}{}\begingroup
  \urlstyle{rm}\Url }%
}%
\providecommand \doibase [0]{http://dx.doi.org/}%
\providecommand \Doi[1]{\href{\doibase#1}}%
\providecommand \bibAnnote [3]{%
  \BibitemShut{#1}%
  \begin{quotation}\noindent
    \textsc{Key:}\ #2\\\textsc{Annotation:}\ #3%
  \end{quotation}%
}%
\providecommand \bibAnnoteFile [2]{%
  \IfFileExists{#2}{\bibAnnote {#1} {#2} {\input{#2}}}{}%
}%
\providecommand \typeout [0]{\immediate \write \m@ne }%
\providecommand \selectlanguage [0]{\@gobble}%
\providecommand \bibinfo [0]{\@secondoftwo}%
\providecommand \bibfield [0]{\@secondoftwo}%
\providecommand \translation [1]{[#1]}%
\providecommand \BibitemOpen[0]{}%
\providecommand \bibitemStop [0]{}%
\providecommand \bibitemNoStop [0]{.\EOS\space}%
\providecommand \EOS [0]{\spacefactor3000\relax}%
\providecommand \BibitemShut [1]{\csname bibitem#1\endcsname}%
\bibitem{Nielsen:2000:CambridgeUniversityPress}%
  \BibitemOpen
  \bibfield{author}{%
  \bibinfo {author} {\bibfnamefont{M.~A.}\ \bibnamefont{Nielsen}}\ and\
  \bibinfo {author} {\bibfnamefont{I.~L.}\ \bibnamefont{Chuang}},\ }%
  \emph{\bibinfo {title} {Quantum Computation and Quantum Information}}\
  (\bibinfo {publisher} {Cambridge University Press},\ \bibinfo {address}
  {Cambridge},\ \bibinfo {year} {2000})%
  \bibAnnoteFile{NoStop}{Nielsen:2000:CambridgeUniversityPress}%
\bibitem{DivencenzoFTPhysRevLett.77.3260}%
  \BibitemOpen
  \bibfield{author}{%
  \bibinfo {author} {\bibfnamefont{D.~P.}\ \bibnamefont{DiVincenzo}}\ and\
  \bibinfo {author} {\bibfnamefont{P.~W.}\ \bibnamefont{Shor}},\ }%
  \bibfield{journal}{%
  \bibinfo {journal} {Phys. Rev. Lett.}\ }%
  \textbf{\bibinfo {volume} {77}},\ \bibinfo {pages} {3260} (\bibinfo {year}
  {1996})%
  \bibAnnoteFile{NoStop}{DivencenzoFTPhysRevLett.77.3260}%
\bibitem{Kitaev:2003:2}%
  \BibitemOpen
  \bibfield{author}{%
  \bibinfo {author} {\bibfnamefont{A.}~\bibnamefont{Kitaev}},\ }%
  \bibfield{journal}{%
  \bibinfo {journal} {Ann. of Phys.}\ }%
  \textbf{\bibinfo {volume} {303}},\ \bibinfo {pages} {2} (\bibinfo {year}
  {2003})%
  \bibAnnoteFile{NoStop}{Kitaev:2003:2}%
\bibitem{QECbook:2013}%
  \BibitemOpen
  \bibfield{author}{%
  \bibinfo {author} {\bibfnamefont{D.}~\bibnamefont{Lidar}}\ and\ \bibinfo
  {author} {\bibfnamefont{T.}~\bibnamefont{Brun}},\ }%
  \emph{\bibinfo {title} {Quantum Error Correction}}\ (\bibinfo {publisher}
  {Cambridge University Press, Cambridge},\ \bibinfo {year} {2013})%
  \bibAnnoteFile{NoStop}{QECbook:2013}%
\bibitem{Zanardi:1999:94}%
  \BibitemOpen
  \bibfield{author}{%
  \bibinfo {author} {\bibfnamefont{P.}~\bibnamefont{Zanardi}}\ and\ \bibinfo
  {author} {\bibfnamefont{M.}~\bibnamefont{Rasetti}},\ }%
  \bibfield{journal}{%
  \bibinfo {journal} {Phys. Lett. A}\ }%
  \textbf{\bibinfo {volume} {264}},\ \bibinfo {pages} {94} (\bibinfo {year}
  {1999})%
  \bibAnnoteFile{NoStop}{Zanardi:1999:94}%
\bibitem{Wilczek:1984:2111}%
  \BibitemOpen
  \bibfield{author}{%
  \bibinfo {author} {\bibfnamefont{F.}~\bibnamefont{Wilczek}}\ and\ \bibinfo
  {author} {\bibfnamefont{A.}~\bibnamefont{Zee}},\ }%
  \bibfield{journal}{%
  \bibinfo {journal} {Phys. Rev. Lett.}\ }%
  \textbf{\bibinfo {volume} {52}},\ \bibinfo {pages} {2111} (\bibinfo {year}
  {1984})%
  \bibAnnoteFile{NoStop}{Wilczek:1984:2111}%
\bibitem{Solina:robustofHQCPhysRevA.70.042316}%
  \BibitemOpen
  \bibfield{author}{%
  \bibinfo {author} {\bibfnamefont{P.}~\bibnamefont{Solinas}}, \bibinfo
  {author} {\bibfnamefont{P.}~\bibnamefont{Zanardi}},\ and\ \bibinfo {author}
  {\bibfnamefont{N.}~\bibnamefont{Zangh\`\i}},\ }%
  \bibfield{journal}{%
  \bibinfo {journal} {Phys. Rev. A}\ }%
  \textbf{\bibinfo {volume} {70}},\ \bibinfo {pages} {042316} (\bibinfo {year}
  {2004})%
  \bibAnnoteFile{NoStop}{Solina:robustofHQCPhysRevA.70.042316}%
\bibitem{Gurdi:robustofHQCPhysRevLett.94.020503}%
  \BibitemOpen
  \bibfield{author}{%
  \bibinfo {author} {\bibfnamefont{I.}~\bibnamefont{Fuentes-Guridi}}, \bibinfo
  {author} {\bibfnamefont{F.}~\bibnamefont{Girelli}},\ and\ \bibinfo {author}
  {\bibfnamefont{E.}~\bibnamefont{Livine}},\ }%
  \bibfield{journal}{%
  \bibinfo {journal} {Phys. Rev. Lett.}\ }%
  \textbf{\bibinfo {volume} {94}},\ \bibinfo {pages} {020503} (\bibinfo {year}
  {2005})%
  \bibAnnoteFile{NoStop}{Gurdi:robustofHQCPhysRevLett.94.020503}%
\bibitem{solinas2012stability_HQC}%
  \BibitemOpen
  \bibfield{author}{%
  \bibinfo {author} {\bibfnamefont{P.}~\bibnamefont{Solinas}}, \bibinfo
  {author} {\bibfnamefont{M.}~\bibnamefont{Sassetti}}, \bibinfo {author}
  {\bibfnamefont{P.}~\bibnamefont{Truini}},\ and\ \bibinfo {author}
  {\bibfnamefont{N.}~\bibnamefont{Zangh{\`\i}}},\ }%
  \bibfield{journal}{%
  \bibinfo {journal} {New. J. Phys}\ }%
  \textbf{\bibinfo {volume} {14}},\ \bibinfo {pages} {093006} (\bibinfo {year}
  {2012})%
  \bibAnnoteFile{NoStop}{solinas2012stability_HQC}%
\bibitem{Farhi:0001106}%
  \BibitemOpen
  \bibfield{author}{%
  \bibinfo {author} {\bibfnamefont{E.}~\bibnamefont{Farhi}}, \bibinfo {author}
  {\bibfnamefont{J.}~\bibnamefont{Goldstone}}, \bibinfo {author}
  {\bibfnamefont{S.}~\bibnamefont{Gutmann}},\ and\ \bibinfo {author}
  {\bibfnamefont{M.}~\bibnamefont{Sipser}},\ }%
  \enquote{\bibinfo {title} {Quantum computation by adiabatic evolution},}\
  (\bibinfo {year} {2000}),\ \bibinfo {note} {eprint arXiv:quant-ph/0001106}%
  \bibAnnoteFile{NoStop}{Farhi:0001106}%
\bibitem{FarhiScience}%
  \BibitemOpen
  \bibfield{author}{%
  \bibinfo {author} {\bibfnamefont{E.}~\bibnamefont{Farhi}}, \bibinfo {author}
  {\bibfnamefont{J.}~\bibnamefont{Goldstone}}, \bibinfo {author}
  {\bibfnamefont{S.}~\bibnamefont{Gutmann}}, \bibinfo {author}
  {\bibfnamefont{J.}~\bibnamefont{Lapan}}, \bibinfo {author}
  {\bibfnamefont{A.}~\bibnamefont{Lundgren}},\ and\ \bibinfo {author}
  {\bibfnamefont{D.}~\bibnamefont{Preda}},\ }%
  \bibfield{journal}{%
  \bibinfo {journal} {Science}\ }%
  \textbf{\bibinfo {volume} {292}},\ \bibinfo {pages} {472} (\bibinfo {year}
  {2001})%
  \bibAnnoteFile{NoStop}{FarhiScience}%
\bibitem{Jordan:2005:052322}%
  \BibitemOpen
  \bibfield{author}{%
  \bibinfo {author} {\bibfnamefont{S.}~\bibnamefont{Jordan}}, \bibinfo {author}
  {\bibfnamefont{E.}~\bibnamefont{Farhi}},\ and\ \bibinfo {author}
  {\bibfnamefont{P.}~\bibnamefont{Shor}},\ }%
  \bibfield{journal}{%
  \bibinfo {journal} {Phys. Rev. A}\ }%
  \textbf{\bibinfo {volume} {74}},\ \bibinfo {pages} {052322} (\bibinfo {year}
  {2006})%
  \bibAnnoteFile{NoStop}{Jordan:2005:052322}%
\bibitem{TameemNJP:adiabaticMarkovianME}%
  \BibitemOpen
  \bibfield{author}{%
  \bibinfo {author} {\bibfnamefont{T.}~\bibnamefont{Albash}}, \bibinfo {author}
  {\bibfnamefont{S.}~\bibnamefont{Boixo}}, \bibinfo {author}
  {\bibfnamefont{D.~A.}\ \bibnamefont{Lidar}},\ and\ \bibinfo {author}
  {\bibfnamefont{P.}~\bibnamefont{Zanardi}},\ }%
  \bibfield{journal}{%
  \bibinfo {journal} {New. J. Phys.}\ }%
  \textbf{\bibinfo {volume} {14}},\ \bibinfo {pages} {123016} (\bibinfo {year}
  {2012})%
  \bibAnnoteFile{NoStop}{TameemNJP:adiabaticMarkovianME}%
\bibitem{OgyanHolonomicQCPhysRevLett.102.070502}%
  \BibitemOpen
  \bibfield{author}{%
  \bibinfo {author} {\bibfnamefont{O.}~\bibnamefont{Oreshkov}}, \bibinfo
  {author} {\bibfnamefont{T.~A.}\ \bibnamefont{Brun}},\ and\ \bibinfo {author}
  {\bibfnamefont{D.~A.}\ \bibnamefont{Lidar}},\ }%
  \bibfield{journal}{%
  \bibinfo {journal} {Phys. Rev. Lett.}\ }%
  \textbf{\bibinfo {volume} {102}},\ \bibinfo {pages} {070502} (\bibinfo {year}
  {2009})%
  \bibAnnoteFile{NoStop}{OgyanHolonomicQCPhysRevLett.102.070502}%
\bibitem{OgyanHolonomicQcPhysRevA.80.022325}%
  \BibitemOpen
  \bibfield{author}{%
  \bibinfo {author} {\bibfnamefont{O.}~\bibnamefont{Oreshkov}}, \bibinfo
  {author} {\bibfnamefont{T.~A.}\ \bibnamefont{Brun}},\ and\ \bibinfo {author}
  {\bibfnamefont{D.~A.}\ \bibnamefont{Lidar}},\ }%
  \bibfield{journal}{%
  \bibinfo {journal} {Phys. Rev. A}\ }%
  \textbf{\bibinfo {volume} {80}},\ \bibinfo {pages} {022325} (\bibinfo {year}
  {2009})%
  \bibAnnoteFile{NoStop}{OgyanHolonomicQcPhysRevA.80.022325}%
\bibitem{LidarTowardFTAdqcPhysRevLett.100.160506}%
  \BibitemOpen
  \bibfield{author}{%
  \bibinfo {author} {\bibfnamefont{D.~A.}\ \bibnamefont{Lidar}},\ }%
  \bibfield{journal}{%
  \bibinfo {journal} {Phys. Rev. Lett.}\ }%
  \textbf{\bibinfo {volume} {100}},\ \bibinfo {pages} {160506} (\bibinfo {year}
  {2008})%
  \bibAnnoteFile{NoStop}{LidarTowardFTAdqcPhysRevLett.100.160506}%
\bibitem{StewartEquivalenceADCPhysRevA.71.062314}%
  \BibitemOpen
  \bibfield{author}{%
  \bibinfo {author} {\bibfnamefont{M.~S.}\ \bibnamefont{Siu}},\ }%
  \bibfield{journal}{%
  \bibinfo {journal} {Phys. Rev. A}\ }%
  \textbf{\bibinfo {volume} {71}},\ \bibinfo {pages} {062314} (\bibinfo {year}
  {2005})%
  \bibAnnoteFile{NoStop}{StewartEquivalenceADCPhysRevA.71.062314}%
\bibitem{AriEquivalenceADCPhysRevLett.99.070502}%
  \BibitemOpen
  \bibfield{author}{%
  \bibinfo {author} {\bibfnamefont{A.}~\bibnamefont{Mizel}}, \bibinfo {author}
  {\bibfnamefont{D.~A.}\ \bibnamefont{Lidar}},\ and\ \bibinfo {author}
  {\bibfnamefont{M.}~\bibnamefont{Mitchell}},\ }%
  \bibfield{journal}{%
  \bibinfo {journal} {Phys. Rev. Lett.}\ }%
  \textbf{\bibinfo {volume} {99}},\ \bibinfo {pages} {070502} (\bibinfo {year}
  {2007})%
  \bibAnnoteFile{NoStop}{AriEquivalenceADCPhysRevLett.99.070502}%
\bibitem{Tanimura2004199}%
  \BibitemOpen
  \bibfield{author}{%
  \bibinfo {author} {\bibfnamefont{S.}~\bibnamefont{Tanimura}}, \bibinfo
  {author} {\bibfnamefont{D.}~\bibnamefont{Hayashi}},\ and\ \bibinfo {author}
  {\bibfnamefont{M.}~\bibnamefont{Nakahara}},\ }%
  \bibfield{journal}{%
  \bibinfo {journal} {Phys, Lett. A}\ }%
  \textbf{\bibinfo {volume} {325}},\ \bibinfo {pages} {199 } (\bibinfo {year}
  {2004})%
  \bibAnnoteFile{NoStop}{Tanimura2004199}%
\bibitem{tanimura:022101}%
  \BibitemOpen
  \bibfield{author}{%
  \bibinfo {author} {\bibfnamefont{S.}~\bibnamefont{Tanimura}}, \bibinfo
  {author} {\bibfnamefont{M.}~\bibnamefont{Nakahara}},\ and\ \bibinfo {author}
  {\bibfnamefont{D.}~\bibnamefont{Hayashi}},\ }%
  \bibfield{journal}{%
  \bibinfo {journal} {J. Math. Phys.}\ }%
  \textbf{\bibinfo {volume} {46}},\ \bibinfo {eid} {022101} (\bibinfo {year}
  {2005})%
  \bibAnnoteFile{NoStop}{tanimura:022101}%
\bibitem{Nakahara:2003:IOP}%
  \BibitemOpen
  \bibfield{author}{%
  \bibinfo {author} {\bibfnamefont{M.}~\bibnamefont{Nakahara}},\ }%
  \emph{\bibinfo {title} {Geometry, Topology and Physics}},\ \bibinfo {edition}
  {2nd}\ ed.\ (\bibinfo {publisher} {Institute of Physics Publishing},\
  \bibinfo {year} {2003})%
  \bibAnnoteFile{NoStop}{Nakahara:2003:IOP}%
\bibitem{Gottesman:9705052}%
  \BibitemOpen
  \bibfield{author}{%
  \bibinfo {author} {\bibfnamefont{D.}~\bibnamefont{{Gottesman}}},\ }%
  \emph{\bibinfo {title} {Stabilizer codes and quantum error correction}},\
  Ph.D. thesis,\ \bibinfo {school} {California Institute of Technology}
  (\bibinfo {year} {1997}),\ \bibinfo {note} {eprint arXiv:quant-ph/9705052}%
  \bibAnnoteFile{NoStop}{Gottesman:9705052}%
\bibitem{KnillFTNature}%
  \BibitemOpen
  \bibfield{author}{%
  \bibinfo {author} {\bibfnamefont{E.}~\bibnamefont{Knill}},\ }%
  \bibfield{journal}{%
  \bibinfo {journal} {Nature (London)}\ }%
  \textbf{\bibinfo {volume} {434}},\ \bibinfo {pages} {39} (\bibinfo {year}
  {2005})%
  \bibAnnoteFile{NoStop}{KnillFTNature}%
\bibitem{Folwer2012PhysRevA.86.032324}%
  \BibitemOpen
  \bibfield{author}{%
  \bibinfo {author} {\bibfnamefont{A.~G.}\ \bibnamefont{Fowler}}, \bibinfo
  {author} {\bibfnamefont{M.}~\bibnamefont{Mariantoni}}, \bibinfo {author}
  {\bibfnamefont{J.~M.}\ \bibnamefont{Martinis}},\ and\ \bibinfo {author}
  {\bibfnamefont{A.~N.}\ \bibnamefont{Cleland}},\ }%
  \bibfield{journal}{%
  \bibinfo {journal} {Phys. Rev. A}\ }%
  \textbf{\bibinfo {volume} {86}},\ \bibinfo {pages} {032324} (\bibinfo {year}
  {2012})%
  \bibAnnoteFile{NoStop}{Folwer2012PhysRevA.86.032324}%
\bibitem{Messiah:1965:North}%
  \BibitemOpen
  \bibfield{author}{%
  \bibinfo {author} {\bibfnamefont{A.}~\bibnamefont{Messiah}},\ }%
  \emph{\bibinfo {title} {Quantum Mechanics, Vol. II}}\ (\bibinfo {publisher}
  {North-Holland Publishing Co.},\ \bibinfo {address} {Amsterdam},\ \bibinfo
  {year} {1965})%
  \bibAnnoteFile{NoStop}{Messiah:1965:North}%
\bibitem{Hagedorn:2002:235}%
  \BibitemOpen
  \bibfield{author}{%
  \bibinfo {author} {\bibfnamefont{G.~A.}\ \bibnamefont{Hagedorn}}\ and\
  \bibinfo {author} {\bibfnamefont{A.}~\bibnamefont{Joye}},\ }%
  \bibfield{journal}{%
  \bibinfo {journal} {J. Math. Anal. and Appl.}\ }%
  \textbf{\bibinfo {volume} {267}},\ \bibinfo {pages} {235} (\bibinfo {year}
  {2002})%
  \bibAnnoteFile{NoStop}{Hagedorn:2002:235}%
\bibitem{lidarAdiabaticaccuracy:102106}%
  \BibitemOpen
  \bibfield{author}{%
  \bibinfo {author} {\bibfnamefont{D.~A.}\ \bibnamefont{Lidar}}, \bibinfo
  {author} {\bibfnamefont{A.~T.}\ \bibnamefont{Rezakhani}},\ and\ \bibinfo
  {author} {\bibfnamefont{A.}~\bibnamefont{Hamma}},\ }%
  \bibfield{journal}{%
  \bibinfo {journal} {J. Math. Phys.}\ }%
  \textbf{\bibinfo {volume} {50}},\ \bibinfo {eid} {102106} (\bibinfo {year}
  {2009})%
  \bibAnnoteFile{NoStop}{lidarAdiabaticaccuracy:102106}%
\bibitem{RydbergsimulatorNatPhysics}%
  \BibitemOpen
  \bibfield{author}{%
  \bibinfo {author} {\bibfnamefont{H.}~\bibnamefont{Weimer}}, \bibinfo {author}
  {\bibfnamefont{M.}~\bibnamefont{Muller}}, \bibinfo {author}
  {\bibfnamefont{I.}~\bibnamefont{Lesanovsky}}, \bibinfo {author}
  {\bibfnamefont{P.}~\bibnamefont{Zoller}},\ and\ \bibinfo {author}
  {\bibfnamefont{H.~P.}\ \bibnamefont{Buchler}},\ }%
  \bibfield{journal}{%
  \bibinfo {journal} {Nat. Phys.}\ }%
  \textbf{\bibinfo {volume} {6}},\ \bibinfo {pages} {382} (\bibinfo {year}
  {2010})%
  \bibAnnoteFile{NoStop}{RydbergsimulatorNatPhysics}%
\bibitem{KempeQuantumGadget}%
  \BibitemOpen
  \bibfield{author}{%
  \bibinfo {author} {\bibfnamefont{J.}~\bibnamefont{Kempe}}, \bibinfo {author}
  {\bibfnamefont{A.}~\bibnamefont{Kitaev}},\ and\ \bibinfo {author}
  {\bibfnamefont{O.}~\bibnamefont{Regev}},\ }%
  \bibfield{journal}{%
  \bibinfo {journal} {SIAM J. Comput.}\ }%
  \textbf{\bibinfo {volume} {35}},\ \bibinfo {pages} {1070} (\bibinfo {year}
  {2006})%
  \bibAnnoteFile{NoStop}{KempeQuantumGadget}%
\bibitem{Oliveira:2008:CQS:2016985.2016987}%
  \BibitemOpen
  \bibfield{author}{%
  \bibinfo {author} {\bibfnamefont{R.}~\bibnamefont{Oliveira}}\ and\ \bibinfo
  {author} {\bibfnamefont{B.~M.}\ \bibnamefont{Terhal}},\ }%
  \bibfield{journal}{%
  \bibinfo {journal} {Quantum Info. Comput.}\ }%
  \textbf{\bibinfo {volume} {8}},\ \bibinfo {pages} {900} (\bibinfo {year}
  {2008})%
  \bibAnnoteFile{NoStop}{Oliveira:2008:CQS:2016985.2016987}%
\bibitem{HQC_adiabatic_realization2013}%
  \BibitemOpen
  \bibfield{author}{%
  \bibinfo {author} {\bibfnamefont{K.}~\bibnamefont{Toyoda}}, \bibinfo {author}
  {\bibfnamefont{K.}~\bibnamefont{Uchida}}, \bibinfo {author}
  {\bibfnamefont{A.}~\bibnamefont{Noguchi}}, \bibinfo {author}
  {\bibfnamefont{S.}~\bibnamefont{Haze}},\ and\ \bibinfo {author}
  {\bibfnamefont{S.}~\bibnamefont{Urabe}},\ }%
  \bibfield{journal}{%
  \bibinfo {journal} {Phys. Rev. A}\ }%
  \textbf{\bibinfo {volume} {87}},\ \bibinfo {pages} {052307} (\bibinfo {year}
  {2013})%
  \bibAnnoteFile{NoStop}{HQC_adiabatic_realization2013}%
\bibitem{NJPNonAdiabaticHolonomic}%
  \BibitemOpen
  \bibfield{author}{%
  \bibinfo {author} {\bibfnamefont{E.}~\bibnamefont{Sj\"{o}vist}}, \bibinfo
  {author} {\bibfnamefont{D.~M.}\ \bibnamefont{Tong}}, \bibinfo {author}
  {\bibfnamefont{L.~M.}\ \bibnamefont{Andersson}}, \bibinfo {author}
  {\bibfnamefont{B.}~\bibnamefont{Hessmo}}, \bibinfo {author}
  {\bibfnamefont{M.}~\bibnamefont{Johansson}},\ and\ \bibinfo {author}
  {\bibfnamefont{K.}~\bibnamefont{Singh}},\ }%
  \bibfield{journal}{%
  \bibinfo {journal} {New J. Phys.}\ }%
  \textbf{\bibinfo {volume} {14}},\ \bibinfo {pages} {103035} (\bibinfo {year}
  {2012})%
  \bibAnnoteFile{NoStop}{NJPNonAdiabaticHolonomic}%
\bibitem{NonAbelianNonAdiabaticHolonomicNature}%
  \BibitemOpen
  \bibfield{author}{%
  \bibinfo {author} {\bibfnamefont{A.~A.}\ \bibnamefont{Abdumalikov~Jr}},
  \bibinfo {author} {\bibfnamefont{J.}~\bibnamefont{Fink}}, \bibinfo {author}
  {\bibnamefont{Juliusson}}, \bibinfo {author} {\bibnamefont{K.}}, \bibinfo
  {author} {\bibfnamefont{M.}~\bibnamefont{Pechal}}, \bibinfo {author}
  {\bibfnamefont{S.}~\bibnamefont{Berger}}, \bibinfo {author}
  {\bibfnamefont{A.}~\bibnamefont{Wallraff}},\ and\ \bibinfo {author}
  {\bibfnamefont{S.}~\bibnamefont{Filipp}},\ }%
  \bibfield{journal}{%
  \bibinfo {journal} {Nature (London)}\ }%
  \textbf{\bibinfo {volume} {496}},\ \bibinfo {pages} {482} (\bibinfo {year}
  {2013})%
  \bibAnnoteFile{NoStop}{NonAbelianNonAdiabaticHolonomicNature}%
\bibitem{FengGuanruPhysRevLett.110.190501}%
  \BibitemOpen
  \bibfield{author}{%
  \bibinfo {author} {\bibfnamefont{G.}~\bibnamefont{Feng}}, \bibinfo {author}
  {\bibfnamefont{G.}~\bibnamefont{Xu}},\ and\ \bibinfo {author}
  {\bibfnamefont{G.}~\bibnamefont{Long}},\ }%
  \bibfield{journal}{%
  \bibinfo {journal} {Phys. Rev. Lett.}\ }%
  \textbf{\bibinfo {volume} {110}},\ \bibinfo {pages} {190501} (\bibinfo {year}
  {2013})%
  \bibAnnoteFile{NoStop}{FengGuanruPhysRevLett.110.190501}%
\bibitem{Bacon-AGTPhysRevLett.103.120504}%
  \BibitemOpen
  \bibfield{author}{%
  \bibinfo {author} {\bibfnamefont{D.}~\bibnamefont{Bacon}}\ and\ \bibinfo
  {author} {\bibfnamefont{S.~T.}\ \bibnamefont{Flammia}},\ }%
  \bibfield{journal}{%
  \bibinfo {journal} {Phys. Rev. Lett.}\ }%
  \textbf{\bibinfo {volume} {103}},\ \bibinfo {pages} {120504} (\bibinfo {year}
  {2009})%
  \bibAnnoteFile{NoStop}{Bacon-AGTPhysRevLett.103.120504}%
\bibitem{Wotton:2012QMemoryReview}%
  \BibitemOpen
  \bibfield{author}{%
  \bibinfo {author} {\bibfnamefont{J.~R.}\ \bibnamefont{Wootton}},\ }%
  \bibfield{journal}{%
  \bibinfo {journal} {J. Mod. Opt.}\ }%
  \textbf{\bibinfo {volume} {59}},\ \bibinfo {pages} {1717} (\bibinfo {year}
  {2012})%
  \bibAnnoteFile{NoStop}{Wotton:2012QMemoryReview}%
\bibitem{chesi2010self-correciton_memory}%
  \BibitemOpen
  \bibfield{author}{%
  \bibinfo {author} {\bibfnamefont{S.}~\bibnamefont{Chesi}}, \bibinfo {author}
  {\bibfnamefont{B.}~\bibnamefont{R{\"o}thlisberger}},\ and\ \bibinfo {author}
  {\bibfnamefont{D.}~\bibnamefont{Loss}},\ }%
  \bibfield{journal}{%
  \bibinfo {journal} {Phy. Rev. A}\ }%
  \textbf{\bibinfo {volume} {82}},\ \bibinfo {pages} {022305} (\bibinfo {year}
  {2010})%
  \bibAnnoteFile{NoStop}{chesi2010self-correciton_memory}%
\bibitem{hutter2012self-correction_memory}%
  \BibitemOpen
  \bibfield{author}{%
  \bibinfo {author} {\bibfnamefont{A.}~\bibnamefont{Hutter}}, \bibinfo {author}
  {\bibfnamefont{J.~R.}\ \bibnamefont{Wootton}}, \bibinfo {author}
  {\bibfnamefont{B.}~\bibnamefont{R{\"o}thlisberger}},\ and\ \bibinfo {author}
  {\bibfnamefont{D.}~\bibnamefont{Loss}},\ }%
  \bibfield{journal}{%
  \bibinfo {journal} {Phy. Rev. A}\ }%
  \textbf{\bibinfo {volume} {86}},\ \bibinfo {pages} {052340} (\bibinfo {year}
  {2012})%
  \bibAnnoteFile{NoStop}{hutter2012self-correction_memory}%
\end{thebibliography}
\end{document}